\documentclass[journal]{IEEEtran}

\usepackage{amssymb}
\usepackage{amsmath}
\usepackage{amsfonts}
\usepackage{multirow}
\usepackage{graphicx}
\usepackage{textcomp}
\usepackage{amsthm}

\usepackage{setspace}
\usepackage{subfigure}
\usepackage{algorithmic}
\usepackage[ruled,vlined,boxed,linesnumbered]{algorithm2e}
\usepackage{color}

\ifCLASSINFOpdf
\else
\fi

\newtheorem{thm}{Theorem}
\newtheorem{lem}{Lemma}

\newtheorem{defn}{Definition}

\newtheorem{exmp}{Example}
\newtheorem{rem}{Remark}

\begin{document}

\title{Composite Community-Aware Diversified Influence Maximization with Efficient Approximation}

\author{Jianxiong Guo,~\IEEEmembership{Member,~IEEE},
	Qiufen Ni,
	Weili Wu,~\IEEEmembership{Senor Member,~IEEE},
	and Ding-Zhu Du
	\thanks{Jianxiong Guo is with the Advanced Institute of Natural Sciences, Beijing Normal University, Zhuhai 519087, China, and also with the Guangdong Key Lab of AI and Multi-Modal Data Processing, BNU-HKBU United International College, Zhuhai 519087, China. (E-mail: jianxiongguo@bnu.edu.cn)
	
	Qiufen Ni is with the School of Computers, Guangdong University of Technology, Guangzhou 510006, China. (E-mail: niqiufen@gdut.edu.cn)
		
	Weili Wu and Ding-Zhu Du are with the Department of Computer Science, Erik Jonsson School of Engineering and Computer Science, The University of Texas at Dallas, Richardson, TX 75080, USA. (E-mail: weiliwu@utdallas.edu; dzdu@utdallas.edu)
	
	\textit{(Corresponding author: Qiufen Ni)}
	}
	\thanks{Manuscript received xxxx; revised xxxx.}}

\markboth{Journal of \LaTeX\ Class Files,~Vol.~xx, No.~xx, September~2022}%
{Shell \MakeLowercase{\textit{et al.}}: Bare Demo of IEEEtran.cls for IEEE Journals}

\maketitle

\begin{abstract}
	Influence Maximization (IM) is a famous topic in mobile networks and social computing, which aims at finding a small subset of users to maximize the influence spread through online information cascade. Recently, some careful researchers paid attention to diversity of information dissemination, especially community-aware diversity, and formulated the diversified IM problem. The diversity is ubiquitous in a lot of real-world applications, but they are all based on a given community structure. In social networks, we can form heterogeneous community structures for the same group of users according to different metrics. Therefore, how to quantify the diversity based on multiple community structures is an interesting question. In this paper, we propose the Composite Community-Aware Diversified IM (CC-DIM) problem, which aims at selecting a seed set to maximize the influence spread and the composite diversity over all possible community structures under consideration.
	
	To address the NP-hardness of CC-DIM problem, we adopt the technique of reverse influence sampling and design a random Generalized Reverse Reachable (G-RR) set to estimate the objective function. The composition of a random G-RR set is much more complex than the RR set used for the IM problem, which will lead to inefficiency of traditional sampling-based approximation algorithms. Because of this, we further propose a two-stage algorithm, Generalized HIST (G-HIST). It can not only return a $(1-1/e-\varepsilon)$ approximate solution with at least $(1-\delta)$ probability, but also improve the efficiency of sampling and ease the difficulty of searching by significantly reducing the average size of G-RR sets. Finally, we evaluate our G-HIST on real datasets against existing algorithms. The experimental results show the effectiveness of our proposed algorithm and its superiority over other baseline algorithms.
\end{abstract}


\IEEEpeerreviewmaketitle

\section{Introduction}
\IEEEPARstart{O}{nline} social networks (OSNs) are connected by hundreds of millions of mobile devices through social media, which has become a popular platform for people to express their views, for companies to promote their products, and for governments to spread their policies. With the rapid development of mobile Internet, there are more than 1.52 billion users active daily on Facebook and 321 million users actively monthly on Twitter, which stimulates the study of the Influence Maximization (IM) problem. It selects a small subset of influential users (seed nodes) in social networks, convinces them to adopt one thing (product, service, or opinion), and utilizes the ``word-of-mouth'' effect to activate other users in social networks through online information cascade. Kempe \textit{et al.} \cite{kempe2003maximizing} formally defined the IM problem as a combinatorial optimization problem, which aims to select a size-$k$ seed set such that the expected number of activated nodes can be maximized in the social network. Subsequently, a series of variant optimization problems based on the IM problem for different real-world applications came into being, such as topic-aware IM \cite{chen2015online} \cite{tian2020deep}, time-aware IM \cite{tong2020time} \cite{guo2021adaptive}, location-aware IM \cite{li2014efficient} \cite{chen2020efficient}, and target-aware IM \cite{guo2019targeted} \cite{cai2020target}.

Both the IM problem and their variant optimization problems were only concerned about how to maximize the total number of activated nodes across the network. They did not care who was the activated node, thus the diversity of activated nodes did not attract enough attention in current research. For example, to consider community-aware diversity, users on social media from different communities usually represent different kinds of people, and their classification metrics can be age, gender, occupation, income, etc. For an organization to advertise its ideas or products, it usually hopes to have diverse followers from different communities, so as to increase its influence more effectively. Besides, the diversity of recommendations is also an important criterion to measure the quality of recommandation systems \cite{yin2019social} \cite{zhang2020personalized}. Thus, the diversity could benefit us as you should not put all your eggs in one basket. To best of our knowledge, only five literatures \cite{tang2014diversified} \cite{zhang2019diversifying} \cite{li2020community} \cite{wang2021efficient} \cite{zhang2021grain} considered the diversity in the IM problem. They all took the community-aware diversity as a specific example and based on the initial work \cite{tang2014diversified} that tried to maximize the expected value of the influence spread and influence diversity.

In the Community-Aware Diversified IM problem, it needs to partition a given social graph into communities in advance, and achieve the diversity in this community structure. A community structure is usually formulated according to a certain metric, which can effectively measure the distance between two nodes. If the matric used to partition a social graph is users' occupations, then the diversity here will be occupation-oriented. However, in real-world applications, there is often more than one metric. Sometimes, we want to consider the diversity based on several metrics at the same time. Let us first look at the following example.
\begin{exmp}\label{exmp1}
	For the government to spread their policies, it can select some influential celebrities to publicize, and spread the influence across the network. At this time, the government should not only allow as many listeners as possible to receive the message, but also consider the diversity, including age, occupation, income, social class, etc.
\end{exmp}
\noindent
Shown as Example \ref{exmp1}, a kind of community structure is obviously not enough, and the previous diversified IM problem cannot cover this scenario based on multiple metrics.

Therefore, we propose a Composite Community-Aware Diversified Influence Maximization (CC-DIM) problem in this paper, which perfectly matches multiple metrics of the community partition. For the CC-DIM problem, we first formulate different community structures according to the metrics we consider, and then select a size-$k$ seed set such that the weighted sum of the expected number of activated nodes and the expected diversity of activated nodes, where the expected diversity is the average over all possible community structures based on different metrics. Then, we prove that the CC-DIM problem is NP-hard, but its objective function is monotone and submodular. Unfortunately, similar to computing the influence spread in the IM problem \cite{chen2010scalable} \cite{chen2010scalable2}, it is \#P-hard as well to compute our objective function, even more difficult. If using the greedy hill-climbing algorithm with Monte Carlo simulations to estimate the objective function, then the computational cost will be unacceptable. In order to improve its computational efficiency, a lot of methods based on the technique of reverse influence sampling (RIS) \cite{borgs2014maximizing} have been utilized to solve the IM problem. Here, we design a novel sampling method, called random Generalized Reverse Reachable (G-RR) set, to unbiasedly estimate the objective function of CC-DIM problem. However, the processes of sampling and searching based on random G-RR sets are much more complex and time-consuming than that based on random RR sets for the IM problem because of the diversity in multiple community structures. Thus, based on random G-RR sets and probabilistic analysis, we further propose a two-stage algorithm, called Generalized HIST (G-HIST), which includes sentinel set selection and remaining set selection. It selects a small-size sentinel set by a small number of random G-RR sets in the first stage, and then utilizes the sentinel set to significantly reduce the average size of random G-RR sets in the second stage. Through detailed theoretical analysis and experimental verification, we prove that the memory consumption and running time will be greatly improved by the G-HIST because of its compressed sampling, and the approximation guarantee will not be affected.

The contributions can be summarized as follows.
\begin{itemize}
	\item To the best of our knowledge, we are the first to consider the diversity on multiple community structures according to different metrics, propose the CC-DIM problem, and prove its hardness, monotonicity, and submodularity.
	\item To tackle the intractability, we design a random G-RR set and the unbiased estimator of the objective function. Then, we propose G-HIST algorithm to further reduce the memory consumption and running time, which can return a $(1-1/e-\varepsilon)$ approximate solution with at least $(1-\delta)$ probability.
	\item We conduct intensive simulations based on real-world social datasets. By comparing our G-HIST with the-state-of-art baselines, the experimental results validate the effectiveness of our proposed sampling and algorithm in approximate performance and efficiency.
\end{itemize}

\textbf{Organization: }In Section \uppercase\expandafter{\romannumeral2}, we summarize the works related this paper. We then introduce our CC-DIM problem and its basic properties in Section \uppercase\expandafter{\romannumeral3}, and sampling techniques used to estimate in Section \uppercase\expandafter{\romannumeral4}. In Section \uppercase\expandafter{\romannumeral5} and \uppercase\expandafter{\romannumeral6}, we elaborate the G-HIST algorithm and its corresponding theoretical analysis. Experiments and discussions are presented in Section \uppercase\expandafter{\romannumeral7}, and finally, Section \uppercase\expandafter{\romannumeral8} concludes this paper.

\section{Related Works}
\textbf{Influence Maximization: }Kempe \textit{et al.} \cite{kempe2003maximizing} first formulated the IM problem and defined it as a combinatorial optimization problem. They proposed two classic diffusion models, Independent Cascade (IC) model and Linear Threshold (LT) model, and proved that the IM problem is NP-hard and the influence spread is monotone and submodular. Given a seed set, it is \#P-hard to compute the influence spread under the IC model \cite{chen2010scalable} and LT-model \cite{chen2010scalable2}. Thus, the greedy hill-climbing algorithm can return a $(1-1/e-\varepsilon)$ approximate solution with Monte Carlo simulations. Borgs \textit{et al.} \cite{borgs2014maximizing} first proposed the technique of RIS to reduce the running time, but they did not give a feasible algorithm and a strict theoretical argument. Subsequently, a plethora of research works focused on following the RIS to further improve efficiency, such as TIM/TIM+ \cite{tang2014influence}, IMM \cite{tang2015influence}, SSA/D-SSA \cite{nguyen2016stop}, and OPIM-C \cite{tang2018online}. Along this line, it could run in $O(k(n+m)\log n/\varepsilon^2)$ expected time and return a $(1-1/e-\varepsilon)$ approximate solution with at least $1-1/n$ probability. Recently, Guo \textit{et al.} \cite{guo2020influence} \cite{guo2022influence} proposed a Hit-and-Stop (HIST) algorithm to tackle the scalability issue in high influence networks by reducing the average size of random RR sets without losing the approximation guarantee. For our CC-DIM problem, we learn from the idea of HIST and formulate our G-HIST algorithm because the size of random G-RR sets will be much larger than that of random RR sets from existing multiple community structures. However, our G-HIST is not a trivial revision of HIST, and we extend to the variant optimization problems in social networks.

\textbf{Community Detection: }Finding communities is a basic task to study large networks. A great deal of researchers tried to mine the underlying community structure by different techniques, such as hierarchical clustering \cite{girvan2002community}, modularity maximization \cite{chen2014community}, statistical inference methods \cite{karrer2011stochastic}, and graph partitioning \cite{shi2000normalized}. However, we only care about the diversity of activated nodes which are usually divided according to users' attributes. Thus, these methods are not suitable for the diversity, but they can be used in other variant problems based on our proposed framework.

\textbf{Diversified Influence Maximization: }Based on the above-mentioned community structures in social networks, the diversity of influence spread has become the inherent demand of viral marketing, which is a typical application of IM problem. Tang \textit{et al.} \cite{tang2014diversified} first defined the diversified IM problem as a combinatorial optimization problem that aims to maximize the weighted sum of the influence spread and diversity, while designing an algorithm to solve it. Zhang \textit{et al.} \cite{zhang2019diversifying} adopted three commonly used utilities in economics to quantify the diversity over communities. Li \textit{et al.} \cite{li2020community} proposed a metric to measure the community-based diversified influence and designed two tree-based heuristic algorithms to reduce the computational cost. Wang \textit{et al.} \cite{wang2021efficient} utilized the IMM algorithm to effectively address the diversified IM problem by sampling with theoretical guarantee. Zhang \textit{et al.} \cite{zhang2021grain} exploited a Graph Neural Networks (GNNs) based method, GRAIN, to combine the diversified IM and greedy algorithm into a unified framework, which significantly improved the efficiency of data selection. However, there is no existing work that considers the multiplicity of diversity in practical applications, and the computational challenge of such a problem has not been effectively address. Solving these two problems is the main contribution of this paper.

\section{Problem Formulation}
In this section, we define our CC-DIM problem from the basic definitions of diffusion model, community, and IM.
\subsection{Preliminaries}
Let $G=(V,E)$ be a directed graph with a node set $V=\{v_1,v_2,\cdots,v_n\}$ and an edge set $E=\{e_1,e_2,\cdots,e_m\}$. In social networks, each node $v\in V$ represent a user and each edge $(u,v)\in E$ represents the relationship, e.g., friendship, between $u$ and $v$. For each edge $(u,v)\in E$, we say that $u$ is the in-neighbor of $v$ and $v$ is the out-neighbor of $u$. For each node $v\in V$, we denote by $N^-(v)$ the set of its in-neighbors and $N^+(v)$ the set of its out-neighbors.

In the information diffusion, we consider that a user is active if she accepts (is activated by) the information cascade from her in-neighobors or she is selected as a seed. The information cascade can be given by a predefined diffusion model, such as the Independent Cascade (IC) model \cite{kempe2003maximizing}. Given a seed set $S\subseteq V$, the IC model  is a discrete-time stochastic cascade process shown as follows: (1) At timestamp $0$, all nodes in $S$ are activated and other nodes in $V\backslash S$ are inactive, where a node keeps active once it is activated; (2) If a node $u$ is activated at timestamp $t$, it has one chance to activate its inactive out-neighbor $v$ with the probabiltiy $p_{uv}$ at timestamp $t+1$, after which it cannot activate any nodes; and (3) The information diffusion terminates when no more inactive nodes can be activated in the subsequent timestamp.

\subsection{Influence Maximization (IM)}
The traditional IM problem is to find a seed set $S\subseteq V$ such that its influence spread $\sigma(S)$, can be maximized, which is the expected number of active nodes after the diffusion terminates. To mathematically define the influence spread, we first introduce a concept called ``realization''. A realization $g=(V,E_g)$, $E_g\subseteq E$, is a subgraph sampled according to the diffusion model. For example, in IC model, each edge $(u,v)\in E$ will be independently contained in $E_g$ with the probability $p_{uv}$. An edge in $E_g$ is called ``live edge'' in realization $g$. Thus, the probability of realization $g$ sampled from $G$ under the IC model is $\Pr[g]=\prod_{e\in E_g}p_e\prod_{e\in E\backslash E_g}(1-p_e)$. Obviously, there are totally $2^m$ possible realizations. The influence cascade on a realization becomes deterministic instead of stochastic process. As a result, the influence spread across the network can be regarded as the expected spread on all possible realizations. Now, the IM problem can be written in an expectation form and formally defined as follows.

\begin{defn}[Influence Maximziation]
	Givne a social graph $G=(V,E)$, a diffusion model (IC model in this paper), and an budget $k$, the IM problem asks to find a seed set $S^\circ$, with at most $k$ nodes, that can maximize the expected influence spread across the graph, i.e.,
	\begin{align}
		S^\circ&\in\arg\max_{|S|\leq k}\sigma(S)\\
		&=\mathbb{E}_{g\sim\mathcal{G}}[|I_g(S)|]=\sum_{g\in\mathcal{G}}\Pr[g]\cdot |I_g(S)|,
	\end{align}
	where $\mathcal{G}$ is the collection of all possible realizations sampled from a given diffusion model and $I_g(S)$ is the node set that contains all nodes can be reached from a node in $S$ by the live edges in the realization $g$.
\end{defn}

Here, the influence function $\sigma$ is a set function. Given a set function $f : 2^V\rightarrow\mathbb{R}_+$ and any two subsets $S$ and $T$ with $S\subseteq T\subseteq V$, we say it is monotone if $f(S)\leq f(T)$ and submodular if $f(S\cup\{v\})-f(S)\geq f(T\cup\{v\})-f(T)$. The IM problem is NP-hard and the influence function is monotone and submodular under the IC model. Under the size constraint, the greedy hill-climbing algorithm of iteratively choosing the node with maximum marginal gain approximates the optimal solution within a factor of $(1-1/e)$ \cite{nemhauser1978analysis}. However, given a seed set $S$, it is \#P-hard to compute the $\sigma(S)$ under the IC model \cite{chen2010scalable}. Thus, the hill-climbing algorithm can only return a $(1-1/e-\varepsilon)$ approximation within the $\Omega(kmn\cdot poly(1/\varepsilon))$ running time through Monte Carlo simulations.

\subsection{Composite Community-Aware Diversified IM}
Generally, there are many communities in any given graph, and the community structure is an essential characteristic of social networks. in this way, the users can be divided into different groups according to a certain metric, and their communication within the group is dense but sparse between groups. Given a social network $G=(V,E)$, we assume that it has a disjoint community structure $\mathcal{C}(G)$ associated with $G$, where $C(G)=\{C_1,C_2,\cdots,C_r\}$ is a partittion of $V$. That is $V=\cup_{i=1}^{r}C_i$ and for any $i,j\in\{1,2,\cdots,r\}$, we have $C_i\cap C_j=\emptyset$. However, when considering the diversified IM problem, the community structure can be partitioned based on different metrics, such as gender, age, race, interest, poor-rich disparity, and consuming behavior, for a variety of real applications. Thus, in the same optimization goal, community structure can be determined by different metrics. We denote by $Q$ the metric set under our consideration. Each element $q\in Q$ is a specific metric that can be used to partition the graph. In each metric $q\in Q$, we define the community structure based on the metric $q$ as $\mathcal{C}_q(G)=\{C^q_1,C^q_2,\cdots,C^q_{r_q}\}$, where the $r_q$ indicates the graph can be divided into $r_q$ communities under the matric $q$. Now, we can formally define the composite diversified function as follows.

\begin{defn}[Composite Diversified Function]
	Given a graph $G=(V,E)$ and a metric set $Q$, the composite diversified function of seed set $S$ is defined as:
	\begin{align}
		\phi(S)&=\sum_{q\in Q}w_q\cdot\psi(S; \mathcal{C}_q)\label{eq2}\\
		&=\sum_{q\in Q}w_q\cdot\sum_{C^q_j\in\mathcal{C}_q}\psi(S; C^q_j)\\
		&=\sum_{q\in Q}w_q\cdot\sum_{C^q_j\in\mathcal{C}_q}\sum_{g\in\mathcal{G}}\Pr[g]\cdot\psi_g(S; C^g_j),\label{eq5}
	\end{align}
	where $w_q$ is the weight of matric $q$ that can ensure $\sum_{q\in Q}w_q=1$, $\psi(S; \mathcal{C}_q)$ is a monotone and submodular set function with respect to $S$ that can quantify the diversity under the community structure generated by the metric $q$, and $\psi(S; C^q_j)$ is the utility under the community $C^q_j\in\mathcal{C}_q$.
\end{defn}

Here, the weight $w_q$ is the importance of diversity of metric $q$, thus the composite diversified function is a weighted average of the influence spread over all kinds of community structures based on multiple different metrics. In order to make the function $\psi(S; \mathcal{C}_q)$ be monotone and submodular with respect to $S$, we should ensure that the $\psi_g(S; C^q_j)$ be monotone and submodular with respect to $S$ because of its linear invariance. In the simplest way, we can let
\begin{equation}\label{psi}
	\psi_g(S; C^q_j)=h_{q,j}\left(\left|I_g(S)\cap C^q_j\right|\right),
\end{equation}
where $h_{q,j}:\mathbb{Z}_+\rightarrow\mathbb{R}_+$ can be any monotone and concave function with $h_{q,j}(0)=0$ since $\psi$ is monotone and submodular function if $h_{q,j}$ is monotone and concave function \cite{lin2011class}. In this paper, we assume that the diversified utility is linear, thus we have $h(x)=a_{q,j}\cdot x$ where $a_{q,j}>0$ is an adjustable coefficient. When we find the proportion of $|I_g(S)\cap C^q_j|/|C^q_j|$ is low, we can enlarge the coefficient $a_{q,j}$ for diversity promotion. This is an important trick to help us achieve the composite diversity. Then, we are enough to define the Composite Community-Aware Diversified Influence Maximization (CC-DIM) problem as follows.
\begin{defn}[CC-DIM]
	Given a graph $G=(V,E)$, a metric set $Q$, and a budget $k$, the CC-DIM problem aims to select a subset $S$ with $|S|\leq k$ that can maximize the following objective function:
	\begin{equation}\label{of}
		f(S)=(1-\lambda)\frac{\sigma(S)}{\sigma(V)}+\lambda\frac{\phi(S)}{\phi(V)},
	\end{equation}
	where $\lambda\in[0,1]$ is an adjustable parameter that can balance the influene spread and the community-aware diversity.
\end{defn}

Through taking the influence spread and diversity over community structures based on different metrics, we can make the IM problem more flexible to adapt to different applications. If caring more about the influence spread, we can make the parameter $\lambda$ approach to zero, and by adjusting the weight $w_q$, we can specify the importance of different ways of community partition. Furthermore, the CC-DIM problem remains at least the same hardness of addressing the IM problem \cite{kempe2003maximizing} and computing the objective function \cite{chen2010scalable}.
\begin{thm}
	The CC-DIM problem is NP-hard and given a seed set $S$, computing the objective function $f(S)$ defined in Eqn. (\ref{of}) is \#P-hard.
\end{thm}
\begin{proof}
	When given the parameter $\lambda=0$, the CC-DIM problem can be reduced to the IM problem and the objective function $f(S)$ can be reduced to the influence function $\sigma(S)$, thus it inherits the NP-hardness of the IM problem and the \#P-hardness of computing the influence function.
\end{proof}
\begin{thm}
	The objective function $f(S)$ defined in Eqn. (\ref{of}) is monotone and submodular with respect to $S$.
\end{thm}
\begin{proof}
	First, the influence function $\sigma(S)$ is monotone and submodular with respect to $S$ under the IC model \cite{kempe2003maximizing} since the $|I_g(S)|$ is monotone and submodular in any realization $g$. Then, it is easy to see that given any realization $g$ and node set $C_j^q$, the $|I_g(S)\cap C_j^q|$ is monotone and submodular as well. In Eqn. (\ref{psi}), we assume the function $h_{q,j}$ is monotone and concave, thus we have $\psi_g(S; C_j^q)$ is monotone and submodular with respect to $S$ because of the composition property proven in \cite{lin2011class}. Naturally, the $\psi(S; \mathcal{C}_q)$ is monotone and submodular because it is a linear combination over all communities $C_j^q\in\mathcal{C}_q(G)$ and realizations $g\in\mathcal{G}$. Thus, we have that the $\phi(S)$ and $f(S)$ are monotone and submodular.
\end{proof}

\section{Sampling Technique}
In the last section, we can argue that the objective function $f(S)$ is monotone and submodular, thus the greedy hill-climbing algorithm can return a solution with $(1-1/e)$ approximation ratio \cite{nemhauser1978analysis}. However, it is \#P-hard to compute this objective function, which leads to the high computational cost to implement it by Monte Carlo simulations. In order to reduce the time complexity and ensure a valid approximation guarantee, we will exploit the the technique of reverse influence sampling (RIS) \cite{borgs2014maximizing} in our problem. This technique is based on the concept of the random reverse reachable (RR) set. In the traditional IM problem, a random RR set $R$ can be generated in three steps: (1) Uniformly select a node $v$ from $V$; (2) Randomly sample a realization $g$ from $\mathcal{G}$ under a given diffusion model; and (3) Collect the nodes that can reach $v$ through a live path in realization $g$. Here, for each node $u\in V$, the probability that it is contained in $R$ randomly generated by the $v$ equals the probability that $u$ can activate $v$. Thus, we have the following lemma.
\begin{lem}[\cite{borgs2014maximizing}]
	Let $S\subseteq V$ be a seed set and $R$ be a random RR set under a given diffusion model, then we have
	\begin{equation}
		\sigma(S)=n\cdot\mathbb{E}_R[\mathbb{I}(S\cap R)],
	\end{equation}
	where $\mathbb{I}(S\cap R)=1$ if $S\cap R\neq\emptyset$, else $\mathbb{I}(S\cap R)=0$.
\end{lem}

\begin{figure*}[!t]
	\begin{align}
		\hat{f}(S; \tilde{\mathcal{R}})&=\frac{1-\lambda}{\sigma(V)}\cdot\hat{\sigma}(S; \tilde{\mathcal{R}})+\frac{\lambda}{\phi(V)}\cdot\hat{\phi}(S; \tilde{\mathcal{R}})\label{eq15}\\
		&=\frac{1-\lambda}{\sigma(V)\theta}\sum_{i=1}^{\theta}\sum_{q\in Q}\frac{1}{|Q|}\sum_{C^q_j\in\mathcal{C}_q}|C^q_j|\cdot\mathbb{I}(S\cap R(v^i_{q,j},g^i))+\frac{\lambda }{\phi(V)\theta}\sum_{i=1}^{\theta}\sum_{q\in Q}w_q\sum_{C^q_j\in\mathcal{C}_q}a_{q,j}\cdot|C^q_j|\cdot\mathbb{I}(S\cap R(v^i_{q,j},g^i))\label{eq16}\\
		&=\frac{1}{\theta}\cdot\underbrace{\sum_{i=1}^{\theta}\sum_{q\in Q}\sum_{C^q_j\in\mathcal{C}_q}\left(\frac{1-\lambda}{\sigma(V)|Q|}+\frac{\lambda w_q a_{q,j}}{\phi(V)}\right)\cdot|C^q_j|\cdot\mathbb{I}(S\cap R(v^i_{q,j},g^i))}_{\Omega(S; \tilde{\mathcal{R}})}.\label{eq17}
	\end{align}
	\hrulefill
	\begin{equation}\label{eq19}
		\Delta_{\hat{f}}(u|S; \tilde{\mathcal{R}})=\frac{1}{\theta}\cdot\sum_{i=1}^{\theta}\sum_{q\in Q}\sum_{C^q_j\in\mathcal{C}_q}\left(\frac{1-\lambda}{\sigma(V)|Q|}+\frac{\lambda w_q a_{q,j}}{\phi(V)}\right)\cdot|C^q_j|\cdot\left[\mathbb{I}((S\cup\{u\})\cap R(v^i_{q,j},g^i))-\mathbb{I}(S\cap R(v^i_{q,j},g^i))\right],
	\end{equation}
	\hrulefill
\end{figure*}

However, this sampling can not be directly applied to address our objective function because we need to consider the diversity under different metrics. To estimate the $\phi(S)$, we need to design a sampling method to compute the value of $\psi(S; C^q_j)$ for each $C_j^q\in\mathcal{C}_q$ and $q\in Q$ as a subroutine. Here, we have the following based on Eqn. (\ref{psi}),
\begin{equation}
	\psi(S; C^q_j)=a_{q,j}\cdot\mathbb{E}_{g\sim\mathcal{G}}[|I_g(S)\cap C^q_j|].
\end{equation}
Denoted by $R(v,g)$ a RR set generated from the node $v$ under the realization $g$, we have
\begin{equation}
	\psi(S; C^q_j)=a_{q,j}\cdot|C_j^q|\cdot\mathbb{E}_{v\sim C^q_j,g\sim\mathcal{G}}[\mathbb{I}(S\cap R(v,g))].
\end{equation}
Then, the composite diversified function defined in Eqn. (\ref{eq5}) can be rewritten as
\begin{equation}\label{eq11}
	\phi(S)=\sum_{g\in\mathcal{G}}\Pr[g]\cdot\sum_{q\in Q}w_q\sum_{C^q_j\in\mathcal{C}_q}\psi_g(S; C^q_j).
\end{equation}

Motivated by Eqn. (\ref{eq11}), all communities $C_j^q$ for $C_j^q\in\mathcal{C}_q$ and $q\in Q$ can share the same realization $g$ in computing the expectation. Thereby, we can formally define the concept of the random Generalized Reverse Reachable (G-RR) set. In our CC-DIM problem, a random G-RR set $\tilde{R}$ can be generated in the following three steps.
\begin{itemize}
	\item Uniformly select a node $v_{q,j}$ from $C_j^q$ for each community $C_j^q\in\mathcal{C}_q$ and $q\in Q$, totally $\sum_{q\in Q}r_q$ nodes.
	\item Randomly sample a realization $g$ from $\mathcal{G}$ under a given diffusion model.
	\item For each node $v_{g,j}$ sampled from $C_j^q$, collect the nodes that can reach it through a live path in $g$. This is a random RR set $R(v_{q,j},g)$, thus we have totally $\sum_{q\in Q}r_q$ random RR set under the same realization.
\end{itemize}
From above, a random G-RR set $\tilde{R}$ is a collection of random RR sets but they are under the same realization, which can be defined as
\begin{equation}\label{eq12}
	\tilde{R}=\{R(v_{q,j},g) : C_j^q\in\mathcal{C}_q \text{ and } q\in Q\}.
\end{equation}
Then, we can build the relationship between the composite deversified function and random G-RR set.

\begin{thm}\label{thm3}
	Let $S\subseteq V$ be a seed set and $\tilde{R}$ be a random G-RR set under a given diffusion model, the function $\sigma(S)$ and $\phi(S)$ can be estimated as follows:
	\begin{small}
	\begin{align}
		&\sigma(S)=\mathbb{E}_{\tilde{R}}\left[\sum_{q\in Q}\frac{1}{|Q|}\sum_{C^q_j\in\mathcal{C}_q}|C^q_j|\cdot\mathbb{I}(S\cap R(v_{q,j},g))\right],\label{eq13}\\
		&\phi(S)=\mathbb{E}_{\tilde{R}}\left[\sum_{q\in Q}w_q\sum_{C^q_j\in\mathcal{C}_q}a_{q,j}\cdot|C^q_j|\cdot\mathbb{I}(S\cap R(v_{q,j},g))\right].\label{eq14}
	\end{align}
	\end{small}
\end{thm}
\begin{proof}
	For a seed set $S$, we have $\mathbb{E}_{g\sim\mathcal{G}}[|I_g(S)\cap C_j^q|]=|C_j^q|\cdot\mathbb{E}_{v\sim C^q_j,g\sim\mathcal{G}}[\mathbb{I}(S\cap R(v,g))]$. It indicates that 
	\begin{align}
		\sigma(S)&=\sum_{C^q_j\in\mathcal{C}_q}\mathbb{E}_{g\sim\mathcal{G}}[|I_g(S)\cap C_j^q|]\nonumber\\
		&=\sum_{C^q_j\in\mathcal{C}_q}|C_j^q|\cdot\mathbb{E}_{v\sim C^q_j,g\sim\mathcal{G}}[\mathbb{I}(S\cap R(v,g))]\nonumber\\
		&=\sum_{q\in Q}\frac{1}{|Q|}\sum_{C^q_j\in\mathcal{C}_q}|C_j^q|\cdot\mathbb{E}_{v\sim C^q_j,g\sim\mathcal{G}}[\mathbb{I}(S\cap R(v,g))]\nonumber\\
		&=\mathbb{E}_{g\sim\mathcal{G}}\left[\sum_{q\in Q}\frac{1}{|Q|}\sum_{C^q_j\in\mathcal{C}_q}|C_j^q|\cdot\mathbb{E}_{v\sim C^q_j}[\mathbb{I}(S\cap R(v,g))]\right]\nonumber\\
		&=\text{Eqn. (\ref{eq13})}.\nonumber
	\end{align}
	Similarly, we have
	\begin{align}
		&\phi(S)=\sum_{q\in Q}w_q\sum_{C^q_j\in\mathcal{C}_q}a_{q,j}\cdot\mathbb{E}_{g\sim\mathcal{G}}[|I_g(S)\cap C^q_j|]\nonumber\\
		&=\sum_{q\in Q}w_q\sum_{C^q_j\in\mathcal{C}_q}a_{q,j}\cdot|C_j^q|\cdot\mathbb{E}_{g\sim\mathcal{G},v\sim C^q_j}[\mathbb{I}(S\cap R(v,g))]\nonumber\\
		&=\mathbb{E}_{g\sim\mathcal{G}}\left[\sum_{q\in Q}w_q\sum_{C^q_j\in\mathcal{C}_q}a_{q,j}\cdot|C_j^q|\cdot\mathbb{E}_{v\sim C^q_j}[\mathbb{I}(S\cap R(v,g))]\right]\nonumber\\
		&=\text{Eqn. (\ref{eq14})}.\nonumber
	\end{align}
	This theorem can be proved.
\end{proof}

Let $\tilde{\mathcal{R}}=\{\tilde{R}_1,\tilde{R}_2,\cdots,\tilde{R}_\theta\}$ be a collection of random G-RR sets that contains $\theta$ random G-RR sets, where in each $\tilde{R}_i\in\tilde{\mathcal{R}}$, we denoted it by $\tilde{R}_i=\{R(v^i_{q,j},g^i) : C_j^q\in\mathcal{C}_q \text{ and } q\in Q\}$ the $i$-th random G-RR set in the collection $\tilde{\mathcal{R}}$. Then, according to Eqn. (\ref{eq13}) and Eqn. (\ref{eq14}) in Theorem \ref{thm3}, the unbiased estimation of our objective function $f(S)$ can be formulated from Eqn. (\ref{eq15}) to Eqn. (\ref{eq17}). It is easy to know that the $\hat{f}(S; \tilde{\mathcal{R}})$ is an unbiased estimation function of $f(S)$ under the collection of G-RR sets $\tilde{\mathcal{R}}$, then $\hat{\sigma}(S; \tilde{\mathcal{R}})$ and $\hat{\phi}(S, \tilde{\mathcal{R}})$ can be defined in a similar way shown in Eqn. (\ref{eq16}). The $\sigma(V)$ and $\phi(V)$ are constants with $\sigma(V)=n$ and $\phi(V)=\sum_{q\in Q}w_q\sum_{C_j^q\in\mathcal{C}_q}a_{q,j}\cdot|C^q_j|$. Thus, the Eqn. (\ref{eq17}) can be further simplified.

\begin{thm}\label{thm4}
	Given a collection of G-RR sets $\tilde{\mathcal{R}}$, the function $\hat{f}(S; \tilde{\mathcal{R}})$ is monotone and submodular with respect to $S$.
\end{thm}
\begin{proof}
	For convenience, we define the marginal gain of $u$ on $S$ as $\Delta_{\hat{f}}(u|S; \tilde{\mathcal{R}})=\hat{f}(S\cup\{u\}; \tilde{\mathcal{R}})-\hat{f}(S; \tilde{\mathcal{R}})$, which can be rewritten as Eqn. (\ref{eq19}) based on the expression of Eqn. (\ref{eq17}). First, it is monotone because $\mathbb{I}(S\cap R(v^i_{q,j},g^i))=1$ implies $\mathbb{I}((S\cup\{u\})\cap R(v^i_{q,j},g^i))=1$, thus we have $\Delta_{\hat{f}}(u|S; \tilde{\mathcal{R}})\geq 0$. Next, we show that it is submodular. Given any $S_1\subseteq S_1\subseteq V$ with $u\notin S_2$, to show $\Delta_{\hat{f}}(u|S_1; \tilde{\mathcal{R}})\geq\Delta_{\hat{f}}(u|S_2; \tilde{\mathcal{R}})$, it is equivalent to prove $\mathbb{I}((S_1\cup\{u\})\cap R(v^i_{q,j},g^i))-\mathbb{I}(S_1\cap R(v^i_{q,j},g^i))\geq\mathbb{I}((S_2\cup\{u\})\cap R(v^i_{q,j},g^i))-\mathbb{I}(S_2\cap R(v^i_{q,j},g^i))$ according to Eqn. (\ref{eq19}). Here, we need to build the connection that we have $\mathbb{I}((S_1\cup\{u\})\cap R(v^i_{q,j},g^i))-\mathbb{I}(S_1\cap R(v^i_{q,j},g^i))=1$ if $\mathbb{I}((S_2\cup\{u\})\cap R(v^i_{q,j},g^i))-\mathbb{I}(S_2\cap R(v^i_{q,j},g^i))=1$, which implies $\mathbb{I}((S_2\cup\{u\})\cap R(v^i_{q,j},g^i))=1$ and $\mathbb{I}(S_2\cap R(v^i_{q,j},g^i))=0$. Obviously, the $\mathbb{I}(S_2\cap R(v^i_{q,j},g^i))=0$ indicates that $S_2\cap R(v^i_{q,j},g^i)=\emptyset$, naturally we have $S_1\cap R(v^i_{q,j},g^i)=\emptyset$ because of $S_1\subseteq S_2$. Then, the $\mathbb{I}((S_2\cup\{u\})\cap R(v^i_{q,j},g^i))=1$ is enough to infer $\{u\}\cap R(v^i_{q,j},g^i)\neq\emptyset$, thus we have $\mathbb{I}((S_1\cup\{u\})\cap R(v^i_{q,j},g^i))=1$. Thereby, we have $\mathbb{I}((S_1\cup\{u\})\cap R(v^i_{q,j},g^i))-\mathbb{I}(S_1\cap R(v^i_{q,j},g^i))=1$ and $\Delta_{\hat{f}}(u|S_1; \tilde{\mathcal{R}})\geq\Delta_{\hat{f}}(u|S_2; \tilde{\mathcal{R}})$.
\end{proof}

\section{Approximation Algorithm}
Based on Theorem \ref{thm3} and Theorem \ref{thm4}, the original problem can be transformed to a weighted maximum coverage problem, whose objective function $\hat{f}(S; \tilde{\mathcal{R}})$ is monotone and submodular given a collection of random G-RR sets $\tilde{\mathcal{R}}$. It is much more convenient and efficient than directly solving the original problem. In this section, we first introduce some preliminary knowledge, and then we begin to design our algorithm and conduct the theoretical analysis step by step.

\begin{algorithm}[h] 
	\caption{{MaxCoverage-Greedy $(\tilde{\mathcal{R}},k)$}}
	\label{a1}
	Initialize: $S_0\leftarrow\emptyset$\;
	\For{$a=1$ to $k$}{
		$v'_a\in\arg\max_{v\in V\backslash S_{a-1}}\{\Delta_\Omega(v|S_{a-1}; \tilde{\mathcal{R}})\}$\;
		$S_a\leftarrow S_{a-1}\cup\{v'_a\}$\;
	}
	\Return $S_k$
\end{algorithm}

\subsection{Preliminary Analysis}
Given a seed set $S$ and a collection of random G-RR sets $\tilde{\mathcal{R}}$, we define the generalized coverage as $\Omega(S; \tilde{\mathcal{R}})=\theta\cdot f(S; \tilde{\mathcal{R}})$ shown as Eqn. (\ref{eq17}). Given a collection of random G-RR sets $\tilde{\mathcal{R}}$, we can apply the MC-Greedy algorithm shown as Algorithm \ref{a1}. It iteratively selects the node $v'_a$ with the maximum marginal coverage $\Delta_\Omega(v|S_{a-1}; \tilde{\mathcal{R}})=\Omega(S_{a-1}\cup\{v\}; \tilde{\mathcal{R}})-\Omega(S_{a-1}; \tilde{\mathcal{R}})$, and returns a set $S_k$ as the final solution. Let $S^*_k$ be the solution returned by MC-Greedy process shown as Algorithm \ref{a1}, $\hat{S}^\circ_k$ be the optimal size-k set that achieves the maximum weighted coverage $\Omega$, and $S^\circ_k$ be the optimal solution of the original objective function $f$. The above MC-Greedy algorithm can guarantee
\begin{equation}\label{eq20}
	\Omega(S^*_k;  \tilde{\mathcal{R}})\geq(1-1/e)\Omega(\hat{S}^\circ_k;  \tilde{\mathcal{R}})\geq(1-1/e)\Omega(S^\circ_k; \tilde{\mathcal{R}}),
\end{equation}
because the $\Omega(S; \tilde{\mathcal{R}})$ is monotone and submodular with respect to $S$. Then, we have the following concentration bound adapted to the martingale analysis in \cite{tang2015influence} \cite{tang2018online}.
\begin{lem}\label{lem2}
	For any $\varepsilon>0$, given a seed set $S$ and a collection of random G-RR sets $\tilde{\mathcal{R}}$, we have
	\begin{align}
		&\Pr\left[\Omega(S; \tilde{\mathcal{R}})\leq(1-\xi)\theta f(S)\right]\leq\exp\left(-\frac{\xi^2\theta f(S)}{2}\right),\label{eq21}\\
		&\Pr\left[\Omega(S; \tilde{\mathcal{R}})\geq(1+\xi)\theta f(S)\right]\leq\exp\left(-\frac{\xi^2\theta f(S)}{2+\frac{2}{3}\xi}\right).\label{eq22}
	\end{align}
\end{lem}
Until now, the most excellent algorithm based on the RIS technique to solve the IM problem is OPIM-C \cite{tang2018online}, where they are optimistic about the selected seed set by the greedy algorithm. Motivated by the idea of OPIM-C, in our CC-DIM problem, we first sample a collection of random G-RR sets $\tilde{\mathcal{R}}_1$ to select a size-$k$ seed set $S^*_k$ in greedy manner as Algorithm \ref{a1} and derive an upper bound $\overline{f}(S^\circ_k)$ of $f(S^\circ_k)$. Second, we sample another collection of random G-RR sets $\tilde{\mathcal{R}}_2$ with $|\tilde{\mathcal{R}}_2|=|\tilde{\mathcal{R}}_1|$ to derive a lower bound $\underline{f}(S^*_k)$ of $f(S^*_k)$. The algorithm will stop until we have
\begin{equation}\label{eq23}
	\underline{f}(S^*_k)/\overline{f}(S^\circ_k)\geq(1-1/e-\varepsilon).
\end{equation}
The tigher these bounds are, the fewer the number of random G-RR sets will be, thus greatly reducing the running time. Based on Lemma 4.2 in \cite{tang2018online}, we can derive the lower bound $\underline{f}(S^*_k)$ under the $\tilde{\mathcal{R}}_2$ with $|\tilde{\mathcal{R}}_2|=\theta_2$ as follows:
\begin{equation}\label{eq24}
	\underline{f}(S^*_k)=\left[\left(\sqrt{\Omega(S^*_k; \tilde{\mathcal{R}}_2)+\frac{2\eta_l}{9}}-\sqrt{\frac{\eta_l}{2}}\right)^2-\frac{\eta_l}{18}\right]\cdot\frac{1}{\theta_2}
\end{equation}
where we have $\eta_l=\ln(1/\delta_l)$ and $\Pr[f(S^*_k)>\underline{f}(S^*_k)]\geq 1-\delta_l$. In a similar way, we can derive the upper bound $\overline{f}(S^\circ_k)$ under the $\tilde{\mathcal{R}}_1$ with $|\tilde{\mathcal{R}}_1|=\theta_1$ as follows:
\begin{equation}\label{eq25}
	\overline{f}(S^\circ_k)=\left(\sqrt{\overline{\Omega}(S^\circ_k; \tilde{\mathcal{R}}_1)+\frac{\eta_u}{2}}+\sqrt{\frac{\eta_u}{2}}\right)^2\cdot\frac{1}{\theta_1}
\end{equation}
where we have $\eta_u=\ln(1/\delta_u)$ and  $\Pr[f(S^\circ_k)<\overline{f}(S^\circ_k)]\geq 1-\delta_u$. Here, the $\overline{\Omega}(S^\circ_k; \tilde{\mathcal{R}}_1)$ is an upper bound of generalized coverage $\Omega(S^\circ_k; \tilde{\mathcal{R}}_1)$, it satisfies $\overline{\Omega}(S^\circ_k; \tilde{\mathcal{R}}_1)\leq\Omega(S^*_k; \tilde{\mathcal{R}}_1)/(1-1/e)$ based on Eqn. (\ref{eq20}). To make it tighter, we can construct the upper bound $\overline{\Omega}(S^\circ_k; \tilde{\mathcal{R}}_1)$ during running the greedy process because of its submodularity. Let $S^*_a$ with $1\leq a\leq k$ be the set of nodes that are selected in the first $a$ iteration in the MC-Greedy algorithm, then an tigher upper bound $\overline{\Omega}(S^\circ_k; \tilde{\mathcal{R}}_1)=$
\begin{equation*}
	\min_{0\leq a\leq k}\left\{\Omega(S^*_a; \tilde{\mathcal{R}}_1)+\sum_{v\in maxMC(S^*_a,k; \tilde{\mathcal{R}}_1)}\Delta_\Omega(v|S^*_a; \tilde{\mathcal{R}}_1)\right\},
\end{equation*}
where $maxMC(S^*_a,k; \tilde{\mathcal{R}}_1)$ be the set of $k$ nodes with the largest marginal coverage in $\tilde{\mathcal{R}}_1$ with respect to $S^*_a$.

\subsection{Generalized HIST Algorithm}
Given a collection of random G-RR sets $\tilde{\mathcal{R}}$, we can quickly get a sub-optimal solution $S^*_k$ by optimizing function $\hat{f}(S; \tilde{\mathcal{R}})$ through the MC-Greedy algorithm. However, how many random G-RR sets do we need to ensure the approximation ratio shown as Eqn. (\ref{eq23}) is unknown. Thus, in this section, we would like to sample enough number of random G-RR sets to achieve an accurate estimation of our objective function and guarantee the approximation. Different from the IM problem, to generate a random RR set, it only needs to uniformly sample a node from the graph. However, to generate a random G-RR set, it needs to uniformly sample a node from every possible community and there are totally $\sum_{q\in Q}r_q$ communities. Thus, a specific challenge in the sampling for our CC-DIM problem is that the size of a random G-RR set is much larger than that of a random RR set. This not only leads to excessive memory usage, but also the running time is significantly increased because of the difficult generation and coverage computing process. To overcome this challenge, a feasible strategy is to reduce the average size of a random G-RR set.

\begin{algorithm}[!t] 
	\caption{{G-HIST $(G, k, \varepsilon, \delta)$}}
	\label{a2}
	Initialize: $\varepsilon_1=\varepsilon_2=\varepsilon/2$, $\delta_1=\delta_2=\delta/2$\;
	$S^*_b=\text{SentinelSet}(G,Q,k,\varepsilon_1,\delta_1)$\;
	$S^*_{k-b}=\text{RemainingSet}(G,Q,k,S^*_b,\varepsilon,\varepsilon_2,\delta_2)$\;
	\Return $S^*_b\cup S^*_{k-b}$
\end{algorithm}

Recently, there is a method to reduce the memory consumption and running time by reducing the average size of random RR sets in the IM problem, called Hit-and-Stop (HIST) \cite{guo2020influence} \cite{guo2022influence}. Motivated by the HIST thought, our solution to the CC-DIM problem is named as Generalized HIST (G-HIST), which can also be divided into the following two stages.
\begin{itemize}
	\item Sentinel set selection: At this stage, we first generate a small number of random G-RR sets, and use it to select a size-$b$ node set $S^*_b$ by the MC-Greedy algorithm, which can guarantee $f(S^*_b)\geq(1-(1-1/k)^b-\varepsilon_1)\cdot f(S^\circ_k)$ with a high probability.
	\item Remaining set selection: At this stage, we need to generate enough number of random G-RR sets to select the remaining size-$(k-b)$ node set $S^*_{k-b}$. But in the generation of a random G-RR set $\tilde{R}$, in each $R(v_{q,j},g)\in\tilde{R}$, the sampling can be terminated if hitting some node in the sentinel set $S^*_b$. Therefore, the cost of generating a random G-RR set can be significantly reduced. Then, it returns $S^*_b\cup S^*_{k-b}$ as the final result and guarantee $f(S^*_b\cup S^*_{k-b})\geq(1-1/e-\varepsilon_1-\varepsilon_2)\cdot f(S^\circ_k)$.
\end{itemize}

From a high level perspective, in the stage of sentinel set selection, at the beginning of the MC-Greedy shown as Algorithm \ref{a1}, the partial solution $S^*_{a-1}$ has a small number of nodes, thus the value of the marginal gain $\Delta_\Omega(v|S^*_{a-1}; \tilde{\mathcal{R}})$ should be very large. Thus, the required number of random G-RR sets to select the node with the maximum marginal gain will be small, and it is easy to provide a $(1-(1-1/k)^b-\varepsilon_1)$ approximate solution. With the foundation of the first stage, the sampling and searching process of the second stage can be accelerated. Then, in the stage of remaining set selection, we will need a greater number of random G-RR sets to select nodes in a greedy manner because the value of the marginal gain is relatively small. The average size of random G-RR sets can be significantly pruned based on the partial solution $S^*_b$ given by the first stage, thus the computational cost is reduced without losing the approximation ratio, where the final result can give a $(1-1/e-\varepsilon_1-\varepsilon_2)$. The G-HIST algorithm can be shown in Algorithm \ref{a2}. Shown as Algorithm \ref{a2}, let $\varepsilon_1=\varepsilon_2=\varepsilon/2$ and $\delta_1=\delta_2=\delta/2$, it can return a $(1-1/e-\varepsilon)$ approximate solution with at least $1-\delta$ probability.

\subsubsection{Sentinel set selection}
A natural question is how to determine the size of sentinel set $S^*_b$. If the size $b$ is too small, it will reduce the hit rate at the second stage, thus weakening the speed-up effect. If the size $b$ is too large, this problem has almost been solved, thus worsening the memory consumption and running time. In other words, the size $b$ should be carefully determined that is able to balance the cost of sampling at the first stage and the speed-up at the second stage. The process of sentinel set selection can be shown in Algorithm \ref{a3}.

\begin{algorithm}[!t] 
	\caption{{SentinelSet $(G, Q, k, \varepsilon_1, \delta_1)$}}
	\label{a3}
	Set $\theta_1=3\cdot\ln(1/\delta_1)$ and $\theta_{max}$ according to Eqn. (\ref{eq27})\;
	Generate a collection of random G-RR sets $\tilde{\mathcal{R}}_1$ with $|\tilde{\mathcal{R}}_1|=\theta$\;
	$i_{max}\leftarrow\lceil\log_2(\theta_{max}/\theta)\rceil$\;
	\Repeat{training stop}{
		server select active users $\mathcal{A}$ uniformly at random, then broadcast $\boldsymbol{\omega}$, $\boldsymbol{\theta}$ and $\hat{p}(y)$ to $\mathcal{A}$;\	

	}
	
	\For{$i=1$ to $i_{max}$}{
		$S^*_k\leftarrow\text{MaxCoverage-Greedy}(\tilde{\mathcal{R}}_1,k)$\;
		Compute the roughly lower bound $\underline{f}'(S^*_a)$ by Eqn. (\ref{eq24}) on $\tilde{\mathcal{R}}_1$ and $S^*_k$, where $1\leq a\leq k$\;
		Get $\overline{f}(S^\circ_k)$ by Eqn. (\ref{eq25}) on $\tilde{\mathcal{R}}_1$, $\delta_u=\delta_1/(3 i_{max})$\;
		Let $b$ be the maximum number such that $\underline{f}'(S^*_a)/\overline{f}(S^\circ_k)\geq(1-(1-1/k)^a-\varepsilon_1)$\;
		Generate a collection of random G-RR sets $\tilde{\mathcal{R}}_2$ with $|\tilde{\mathcal{R}}_2|=|\tilde{\mathcal{R}}_1|$ by calling G-RR Set-Sentinel $(G,Q,S^*_b)$ shown as Algorithm \ref{a4}\;
		Get $\underline{f}(S^*_b)$ by Eqn. (\ref{eq24}) on $\tilde{\mathcal{R}}_2$, $\delta_l=\delta_1/(6i_{max})$\;
		\If{$\underline{f}(S^*_b)/\overline{f}(S^\circ_k)\geq(1-(1-1/k)^b-\varepsilon_1)$}{
			\Return $S^*_b$\;
		}
		Enlarge $\tilde{\mathcal{R}}_2$ until $|\tilde{\mathcal{R}}_2|=4\cdot|\tilde{\mathcal{R}}_1|$ and re-compute $\underline{f}(S^*_b)$ by Eqn. (\ref{eq24}) on $\tilde{\mathcal{R}}_2$\;
		\If{$\underline{f}(S^*_b)/\overline{f}(S^\circ_k)\geq(1-(1-1/k)^b-\varepsilon_1)$}{
			\Return $S^*_b$\;
		}
		Double the size of $\tilde{\mathcal{R}}_1$\;
	}
	\Return $S^*_b$\;
\end{algorithm}

Shown as Algorithm \ref{a3}, we first give a collection of random G-RR set $\tilde{\mathcal{R}}_1$, and use it to generate a size-$k$ seed set $S^*_k$ by the MC-Greedy algorithm shown in Algorithm \ref{a2}. In this process, we simultaneously obtain the partial solution $S^*_a$ with $1\leq a\leq k$, which can be applied to compute the upper bound $\overline{f}(S^\circ_k)$ by Eqn. (\ref{eq25}). However, based on Eqn. (\ref{eq24}), we need another collection of random G-RR set $\tilde{\mathcal{R}}_2$, which is independently sampled, to compute the lower bound $\underline{f}(S^*_a)$. Let us ignore this point for the time being, where we still apply $\tilde{\mathcal{R}}_1$ to roughly compute the lower bound, denoted by $\underline{f}'(S^*_a)$ to discriminate, for all $1\leq a\leq k$. Then, in line 8 of Algorithm \ref{a3}, we select the maximum $a$, denoted by $b$, such that $\underline{f}'(S^*_a)/\overline{f}(S^\circ_k)\geq(1-(1-1/k)^a-\varepsilon_1)$. Since the roughly lower bound $\underline{f}'(S^*_b)$ may be not accurate, we generate another collection of random G-RR set $\tilde{\mathcal{R}}_2$ and use it to compute the lower bound $\underline{f}(S^*_b)$ by Eqn. (\ref{eq24}). Whereby, we can check whether the $S^*_b$ is at least $(1-(1-1/k)^b-\varepsilon_1)$ approximation. If yes, return the $S^*_b$ directly; If no, make the lower bound tighter through enlarging $\tilde{\mathcal{R}}_2$ until $|\tilde{\mathcal{R}}_2|=4\cdot|\tilde{\mathcal{R}}_1|$ and use it to compute the $\underline{f}(S^*_b)$ again. If the $S^*_b$ can provide the approximation, return it directly; If not, this implies the $S^*_b$ is not a good solution with a high probability. Thus, we double the collection $\tilde{\mathcal{R}}_1$ and repeat the above process to re-select a node set until satisfying the approximation ratio or reaching the maximum number of iterations $i_{max}$.

\begin{algorithm}[!t] 
	\caption{{G-RR Set-Sentinel $(G, Q, S^*_b)$}}
	\label{a4}
	Initialize $\tilde{R}$ as a map, and $(q,j)$ is the key for $q\in Q$ and $C^q_j\in\mathcal{C}_q(G)$\;
	Sample a realization $g$ from $\mathcal{G}$ randomly\;
	\ForEach{$q\in Q$}{
		\ForEach{$C^q_j\in\mathcal{C}_q(G)$}{
			Select a node $v_{q,j}$ from $C^q_j$ uniformly\;
			\If{$v_{q,j}\in S^*_b$}{
				$\tilde{R}[(q,j)]\leftarrow\square$\;
				Continue\;
			}
			Initalize a set $R\leftarrow\emptyset$ and a queue $H\leftarrow\emptyset$\;
			$R\leftarrow R\cup\{v_{q,j}\}$\;
			$H\leftarrow H\cup\{v_{q,j}\}$; Mark $v_{q,j}$ as activated\;
			$Flag\leftarrow False$\;
			\While{$H$ is not empty}{
				Let $u$ be the top node of $H$, pop $u$ from $H$\;
				\ForEach{in-neighbor $w$ of $u$ in $g$}{
					\If{$w$ is inactivated}{
						\If{$w\in S^*_b$}{
							$Flag\leftarrow True$\;
							Break\;
						}
						$R\leftarrow R\cup\{w\}$\;
						$H\leftarrow H\cup\{w\}$; Mark $w$ as activated\;
					}
				}
				\If{Flag}{
					Break\;
				}
			}
			\If{Flag}{
				$\tilde{R}[(q,j)]\leftarrow\square$\;
			} \Else {
				$\tilde{R}[(q,j)]\leftarrow R$\;
			}
		}
	}
	\Return $\tilde{R}$\;
\end{algorithm}

It is worth noting that in line 9 of Algorithm \ref{a3}, the only purpose of the collection of G-RR set $\tilde{\mathcal{R}}_2$ is to compute the lower bound $\underline{f}(S^*_b)$ for a fixed node set $S^*_b$. Given a sentinel set $S^*_b$, the sampling process of a random G-RR set can be optimized, which is shown in Algorithm \ref{a4}. Here, the sampling of collection and searching process of subsequent coverage computation can be significantly improved, which will be widely used in the next stage. Shown as Algorithm \ref{a4}, we elaborate the process of generating a random G-RR set with the help of sentinel set $S^*_b$. First, we initialize a map $\tilde{R}$ (a data structure to represent a G-RR set), where the value of $\tilde{R}[(q,j)]$ contains a random RR set generated from community $C^j_q$, representing the same meaning as $R(v_{q,j},g)$ in Eqn. (\ref{eq12}). For each community $C_j^q\in\mathcal{C}_q$ and $q\in Q$, under the realization $g$, we first select a node $v_{q,j}$ from $C_j^q$ uniformly. If the $v_{q,j}$ hit the sentinel set $S^*_b$, we set $\tilde{R}[(q,j)]$ by $\square$ (a placeholder, which means that the RR set  $\tilde{R}[(q,j)]$ has been covered), then terminate the current iteration and enter the next community. If the $v_{q,j}$ does not hit the $S^*_b$, we add $v_{q,j}$ into the set $R$ and queue $H$, and start a traversal from $v_{q,j}$ following the reverse direction of its edges in the while loop from line 13 to line 23. Here, we use a flag in line 12 to indicate whether the sampling hits the $S^*_b$. If yes, the flag will become true, and we set $\tilde{R}[(q,j)]$ by $\square$; If no, the flag will keep false, and we set $\tilde{R}[(q,j)]$ by the sampled RR set $R$. Given a collection $\tilde{\mathcal{R}}_2$, when computing the value of $\underline{f}(S^*_b)$, we need to compute the value of $\Omega(S^*_b,\tilde{\mathcal{R}}_2)$ as Eqn. (\ref{eq17}). Thus, we have $\mathbb{I}(S^*_b\cap R(v^i_{q,j},g^i))=1$ if and only if $\tilde{R}_i[(q,j)]=\square$, which is much easier to compute than before.

\begin{algorithm}[!t] 
	\caption{{RemainingSet $(G, Q, k, S^*_b,\varepsilon, \varepsilon_2, \delta_2)$}}
	\label{a5}
	Set $\theta_1=3\cdot\ln(1/\delta_2)$ and $\theta_{max}$ according to Eqn. (\ref{eq27})\;
	Generate two collections of random G-RR sets $\tilde{\mathcal{R}}_1$ and $\tilde{\mathcal{R}}_2$ with $|\tilde{\mathcal{R}}_1|=|\tilde{\mathcal{R}}_2|=\theta$ by calling G-RR Set-Sentinel $(G,Q,S^*_b)$ shown as Algorithm \ref{a4}\;
	$i_{max}\leftarrow\lceil\log_2(\theta_{max}/\theta_1)\rceil$\;
	\For{$i=1$ to $i_{max}$}{
		Select a size-$(k-b)$ node set $S^*_{k-b}$ from $V\backslash S^*_b$ based on $\tilde{\mathcal{R}}_1$ by the MC-Greedy algorithm\;
		$S^*_k\leftarrow S^*_b\cup S^*_{k-b}$\;
		Get $\overline{f}(S^\circ_k)$ by Eqn. (\ref{eq25}) on $\tilde{\mathcal{R}}_1$, $\delta_u=\delta_2/(3i_{max})$\;
		Get $\underline{f}(S^*_k)$ by Eqn. (\ref{eq24}) on $\tilde{\mathcal{R}}_2$, $\delta_l=\delta_2/(3i_{max})$\;
		\If{$\underline{f}(S^*_b)/\overline{f}(S^\circ_k)\geq(1-1/e-\varepsilon)$}{
			\Return $S^*_{k-b}$\;
		}
		Double the size of $\tilde{\mathcal{R}}_1$ and $\tilde{\mathcal{R}}_2$ by Algorithm \ref{a4}\;
	}
	\Return $S^*_{k-b}$\;
\end{algorithm}

Next, how many random G-RR sets are enough in the collection $\tilde{\mathcal{R}}_1$ to generate a sentinel set $S^*_b$ with a good approximation? Similar to Lemma 6 in HIST \cite{guo2020influence}, we can give the following theorem.
\begin{thm}\label{thm5}
	Let $\tilde{\mathcal{R}}_1$ be the collection of random G-RR sets and $S^*_b$ be a size-$b$ node set selected by Algorithm \ref{a1} based on $\tilde{\mathcal{R}}_1$. Given any $\varepsilon'$ and $\delta'$, if the size of $\tilde{\mathcal{R}}_1$ satisfies $|\tilde{\mathcal{R}}_1|\geq$
	\begin{small}
	\begin{equation}\label{eq26}
		\frac{2\cdot\left[\left(1-(\frac{1}{k})^b\right)\sqrt{\ln\frac{2}{\delta'}}+\sqrt{\left(1-(\frac{1}{k})^b\right)\left(\ln\binom{n}{b}+\ln\frac{2}{\delta'}\right)}\right]^2}{\varepsilon'^2\cdot f(S^\circ_k)},
	\end{equation}
	\end{small}then we have $f(S^*_b)\geq(1-(1-1/k)^b-\varepsilon')\cdot f(S^\circ_k)$ with at least $1-\delta'$ probability.
\end{thm}
\noindent
Based on Theorem \ref{thm5}, we need to give a lower bound of the $f(S^\circ_k)$ to get a $\theta_{max}$. Here, we define $f_{min}=$
\begin{equation*}
	\sup_{S_k\subseteq V}\{f(S_k)\}=(1-\lambda)\cdot\frac{k}{n}+\lambda\cdot\frac{\min_{q\in Q, C^q_j\in\mathcal{C}_q}\{a_{q,j}\}\cdot k}{\sum_{q\in Q}w_q\sum_{C_j^q\in\mathcal{C}_q}a_{q,j}\cdot|C^q_j|}.
\end{equation*} 
By replacing $f(S^\circ_k)$ with $f_{min}$, $\ln\binom{n}{b}$ with $\ln\binom{n}{k}$, $1-(1/k)^b$ with $1$, and setting $\varepsilon'=\varepsilon_1$ and $\delta'=\delta_1/3$, the maximum number of random G-RR sets in the stage of sentinel set selection is
\begin{equation}\label{eq27}
	\theta_{max}=\frac{2\cdot\left(\sqrt{\ln\frac{6}{\delta_1}}+\sqrt{\ln\binom{n}{k}+\ln\frac{6}{\delta_1}}\right)^2}{\varepsilon_1^2\cdot f_{min}}.
\end{equation}
Thus, if the size of the collection $\tilde{\mathcal{R}}_1$ is larger than $\theta_{max}$, the node set $S^*_b$ selected based on $\tilde{\mathcal{R}}_1$ satisfies $(1-(1-1/k)^b-\varepsilon_1)$ approximation with at least $1-\delta_1/3$ probability.

Shown as Algorithm \ref{a3}, it has at most $i_{max}$ iterations. In the last iteration, if the $\underline{f}(S^*_b)/\overline{f}(S^\circ_k)$ is still unqualified, it will return $S^*_b$ as the final result, where the failure probability, i.e. $f(S^*_b)<(1-(1-1/k)^b-\varepsilon_1)\cdot f(S^\circ_b)$, is at most $\delta_1/3$. In each of the first $i_{max}-1$ iterations, the failure probability of the upper bound in line 7 is $\delta_1/(3i_{max})$ and the failure probabilities of the lower bound in line 11 and 14 are $\delta_1/(6i_{max})$ respectively. By the union bound, the total failure probability of the first $i_{max}-1$ iterations is at most $2\delta_1/3$, then the sentinel set returned by Algorithm \ref{a3} satisfies the desired approximation guarantee with at least $1-\delta_1$ probability.

\subsubsection{Remaining set selection}
After obtaining the sentinel set $S^*_b$ at the first stage, we make full use of it to accelerate the generation of random G-RR sets and get the remaining $k-b$ seed nodes. The process of remaining set selection can be shown in Algorithm \ref{a5}. Here, we first sample two collections of random RR set $\tilde{\mathcal{R}}_1$ and $\tilde{\mathcal{R}}_2$ by invoking Algorithm \ref{a4}. Based on $\tilde{\mathcal{R}}_1$, we can select a size-$(k-b)$ node set $S^*_{k-b}$ from $V\backslash S^*_b$ by the MC-Greedy algorithm.
\begin{rem}\label{rem1}
	In line 5 of Algorithm \ref{a5}, given a collection $\tilde{\mathcal{R}}_1$, we can apply the MC-Greedy algorithm to iteratively select the optimal node. But there is a difference here since it is a greedy strategy based on $S^*_{b}$. Thus, according to Algorithm \ref{a1}, we make a little change. We initialize $S_0\leftarrow S^*_{b}$, and iteratively select from $a=1$ to $k-b$, finally return $S_{k-b}$. Thus, when computing the value of $\Omega(S_{a-1}, \tilde{\mathcal{R}}_1)$ as Eqn. (\ref{eq17}) in the MC-Greedy process, we have $\mathbb{I}(S_{a-1}\cap R(v^i_{q,j},g^i))=1$ if and only if $S_{a-1}\cap\tilde{R}_i[(q,j)]\neq\emptyset$ or $\tilde{R}_i[(q,j)]=\square$. This is also the core mystery of our G-HIST algorithm.
\end{rem}
\noindent
After obtaining a feasible solution $S^*_k$ in line 6, we use the $\tilde{\mathcal{R}}_1$ to compute the upper bound $\overline{f}(S^\circ_k)$ and use the $\tilde{\mathcal{R}}_2$ to compute the lower bound $\underline{f}(S^*_k)$. If the $\underline{f}(S^*_b)/\overline{f}(S^\circ_k)$ has satisfies a $(1-1/e-\varepsilon)$ approximation, we return the $S^*_k$ directly; Otherwise, we double the collection $\tilde{\mathcal{R}}_1$ and $\tilde{\mathcal{R}}_2$, and repeat the above process to re-select the remaining node set until satisfying the approximation ratio or reaching the maximum number of iteration $i_{max}$.

According to Remark \ref{rem1}, in the stage of remaining set selection, the average size of random G-RR sets can be significantly reduced and the computational process of coverage in the MC-Greedy algorithm can also be simplified because the RR-set that intersects the sentinel set $S^*_b$ has been discharged in advanced. Next, how many random G-RR sets are enough in the collection $\tilde{\mathcal{R}}_1$ to generate a remaining set $S^*_{k-b}$ with a good approximation guarantee? Similar to Lemma 7 in HIST \cite{guo2020influence}, we can give the following theorem.
\begin{thm}\label{thm6}
	Given any $\varepsilon'$, $\delta'$, and $S^*_b$ with $f(S^*_b)\geq(1-(1-1/k)^b-\varepsilon_1)\cdot f(S^\circ_k)$, if the size of $\tilde{\mathcal{R}}_1$ satisfies $|\tilde{\mathcal{R}}_1|\geq$
	\begin{equation}\label{eq28}
		\frac{2\cdot\left[	\sqrt{\ln\frac{3}{\delta'}}+\sqrt{\left(1-\frac{1}{e}\right)\left(\ln\binom{n-b}{k-b}+\ln\frac{3}{\delta'}\right)}\right]^2}{\varepsilon'^2\cdot f(S^\circ_k)},
	\end{equation}
	then the remaining set $S^*_{k-b}$ selected by the adapted MC-Greedy algorithm satisfies $f(S^*_b\cup S^*_{k-b})\geq(1-1/e-\varepsilon_1-\varepsilon')\cdot f(S^\circ_k)$ with at least $1-\delta'$ probability.
\end{thm}
\noindent
Based on Theorem \ref{thm6}, by replacing $f(S^\circ_k)$ with $f_{min}$, and setting $\varepsilon'=\varepsilon_2$ and $\delta'=\delta_2/3$, the maximum number of random G-RR sets in the stage of remaining set selection is

\begin{small}
\begin{equation}\label{eq29}
	\theta_{max}=\frac{2\cdot\left[\sqrt{\ln\frac{9}{\delta_2}}+\sqrt{\left(1-\frac{1}{e}\right)\left(\ln\binom{n-b}{k-b }+\ln\frac{9}{\delta_2}\right)}\right]^2}{\varepsilon_2^2\cdot f_{min}}.
\end{equation}
\end{small}Thus, if the size of the collection $\tilde{\mathcal{R}}_1$ is larger than $\theta_{max}$, the node set $S^*_{k-b}$ selected based on $\tilde{\mathcal{R}}_1$ satisfies $(1-1/e-\varepsilon_1-\varepsilon_2)$ approximation with at least $1-\delta_2/3$ probability.

Shown as Algorithm \ref{a5}, if the $\underline{f}(S^*_b)/\overline{f}(S^\circ_k)$ is still unqualified in the last iteration, it will return $S^*_{k-b}$ as the final result with at most $\delta_2/3$ failure probability. By the union bound, the total failure probability of the first $i_{max}-1$ iterations is at most $2\delta_2/3$, then the remaining set returned by Algorithm \ref{a5} satisfies the desired approximation guarantee with at least $1-\delta_2$ probability.

\section{Theoretical Analysis}
In this section, we first prove the Theorem \ref{thm5} and Theorem \ref{thm6} proposed in last section, then show the analysis of time complexity and main theoretical result in this paper.
\subsection{Proof of Theorem \ref{thm5}}
We first give several lemmas as follows that are useful to prove the Theorem \ref{thm5}.
\begin{lem}\label{lem3}
	Let $S^*_b$ be the node set selected by Algorithm \ref{a1} on $\tilde{\mathcal{R}}$, then we have $\Omega(S^*_b; \tilde{\mathcal{R}})\geq (1-(1-1/k)^b)\cdot \Omega(S^\circ_k; \tilde{\mathcal{R}})$.
\end{lem}
\noindent
This lemma is directly from the monotonicity and submodularity of  $\Omega(S; \tilde{\mathcal{R}})$ with respect to $S$ shown in Theorem \ref{thm4}. Denoting by $|\tilde{\mathcal{R}}|=\theta$, the $\Omega(S^\circ_k; \tilde{\mathcal{R}})/\theta$ is an unbiased estimation of $f(S^\circ_k)$. Thus, we have $\Omega(S^\circ_k; \tilde{\mathcal{R}})/\theta\approx f(S^\circ_k)$ when the $\theta$ is large enough.
\begin{lem}\label{lem4}
	Given $\delta'_1$, $\varepsilon'_1$, and 
	\begin{equation}\label{eq30}
		\theta_1=2\cdot\ln(1/\delta'_1)/({\varepsilon'_1}^2\cdot f(S^\circ_k)),
	\end{equation}
	if $\theta\geq\theta_1$, then we have $\Omega(S^\circ_k; \tilde{\mathcal{R}})/\theta> (1-\varepsilon'_1)\cdot f(S^\circ_k)$ with at least $1-\delta'_1$ probability.
\end{lem}
\begin{proof}
	By applying Eqn. (\ref{eq21}) in Lemma \ref{lem2}, we have $\Pr[\Omega(S^\circ_k; \tilde{\mathcal{R}})/\theta\leq(1-\varepsilon'_1)\cdot f(S^\circ_k)]\leq\exp(-({\varepsilon'_1}^2/2)\cdot\theta f(S^\circ_k))\leq \exp(-({\varepsilon'_1}^2/2)\cdot\theta_1 f(S^\circ_k))=\delta'_1$ by substituting $\theta_1$ with the above Eqn. (\ref{eq30}).
\end{proof}
\noindent
Based on Lemma \ref{lem3} and Lemma \ref{lem4}, we have
\begin{equation}\label{eq31}
	\Omega(S^*_b; \tilde{\mathcal{R}})/\theta>(1-(1-1/k)^b)(1-\varepsilon'_1)\cdot f(S^\circ_k)
\end{equation}
with at least $1-\delta'_1$ probability. Next, we can connect the $f(S^*_b)$ with $f(S^\circ_k)$ as the following lemma.
\begin{lem}\label{lem5}
	Given $\delta'_2$, $\varepsilon'$ with $\varepsilon'>\varepsilon'_1$, and
	\begin{equation}
		\theta_2=\frac{2\cdot\left[1-(1-1/k)^b\right]\cdot\left(\ln\binom{n}{b}+\ln\frac{1}{\delta'_2}\right)}{[\varepsilon'-(1-(1-1/k)^b)\cdot\varepsilon'_1]^2\cdot f(S^\circ_k)},
	\end{equation}
	if Eqn. (\ref{eq31}) holds and $\theta\geq\theta_2$, then we have $f(S^*_b)>(1-(1-1/k)^b-\varepsilon')\cdot f(S^\circ_k)$ with at least $1-\delta'_2$ probability.
\end{lem}
\begin{proof}
	Let $S_b$ be any size-$b$ subset of $V$. By applying Eqn. (\ref{eq22}) in Lemma \ref{lem2} and setting $\varepsilon'_2=\varepsilon'-(1-(1-1/k)^b)\cdot\varepsilon'_1$, we have
	\begin{align}
		\Pr&\left[\Omega(S_b; \tilde{\mathcal{R}})/\theta-f(S_b)\geq\varepsilon'_2\cdot f(S^\circ_k)\right]\nonumber\\
		&=\Pr\left[\Omega(S_b; \tilde{\mathcal{R}})-\theta f(S_b)\geq\frac{\varepsilon'_2\cdot f(S^\circ_k)}{f(S_b)}\cdot \theta f(S_b)\right]\nonumber\\
		&\leq\exp\left(-\frac{{\varepsilon'_2}^2\cdot f(S^\circ_k)^2}{2f(S_b)+\frac{2}{3}\varepsilon'_2\cdot f(S^\circ_k)}\cdot\theta\right)\nonumber\\
		&\leq\exp\left(-\frac{{\varepsilon'_2}^2\cdot f(S^\circ_k)}{2(1-(1-1/k)^b-\varepsilon')+\frac{2}{3}\varepsilon'_2}\cdot\theta\right)\nonumber\\
		&\leq\exp\left(-\frac{{\varepsilon'_2}^2\cdot f(S^\circ_k)}{2(1-(1-1/k)^b)}\cdot\theta_2\right)\nonumber\\
		&=\delta'_2/\binom{n}{b}.\nonumber
	\end{align}
	Thus, if $\theta\geq\theta_2$, the $S^*_b$ returned by Algorithm \ref{a1} satisfies $\Omega(S^*_b; \tilde{\mathcal{R}})/\theta-f(S^*_b)<\varepsilon'_2\cdot f(S^\circ_k)$ with at least $1-\delta'_2$ probability based on the union bound of at most $\binom{n}{b}$ size-$b$ node sets. Thus, if the Eqn. (\ref{eq31}) holds, we have
		\begin{align}
		f(S^*_b)&>\Omega(S^*_b; \tilde{\mathcal{R}})/\theta-\varepsilon'_2\cdot f(S^\circ_k)\nonumber\\
		&>[(1-(1-1/k)^b)(1-\varepsilon'_1)-\varepsilon'_2]\cdot f(S^\circ_k)\nonumber\\
		&=(1-(1-1/k)^b-\varepsilon')\cdot f(S^\circ_k)\nonumber
	\end{align}
	with at least $1-\delta'_2$ probability.
\end{proof}

Based on Lemma \ref{lem4} and Lemma \ref{lem5}, if $\theta\geq\max\{\theta_1,\theta_2\}$, we have $f(S^*_b)>(1-(1-1/k)^b-\varepsilon')\cdot f(S^\circ_k)$ with at least $1-\delta'_1-\delta'_2$ probability. By setting $\delta'_1=\delta'_2=\delta'/2$ and $\theta_1=\theta_2=\theta'$, we have $\theta'$ equals the Eqn. (\ref{eq26}) similar to the techniques in \cite{tang2015influence}. Theorem \ref{thm5} has been proven.

\subsection{Proof of Theorem \ref{thm6}}
We first give several lemmas as follows that are useful to prove the Theorem \ref{thm6}.
\begin{lem}\label{lem6}(\cite{guo2020influence})
	Given any size-$b$ sentinel set $S^*_b$, let $S^*_{k-b}$ be a size-$(k-b)$ remaining set selected by the adapted MC-Greedy algorithm from $V\backslash S^*_b$ on $\tilde{\mathcal{R}}$ like the process of Remark \ref{rem1}. We have $\Omega(S^*_b\cup S^*_{k-b}; \tilde{\mathcal{R}})\geq$
	\begin{equation}
		(1-(1-1/k)^{k-b})\cdot\Omega(S^\circ_k; \tilde{\mathcal{R}})+(1-1/k)^{k-b}\cdot\Omega(S^*_b; \tilde{\mathcal{R}}).
	\end{equation}
\end{lem}
\noindent
This lemma is directly from Lemma 11 in \cite{guo2020influence}. Shown as Lemma \ref{lem4}, when $\theta\geq\theta_1$, the $\Omega(S^\circ_k; \tilde{\mathcal{R}})/\theta$ should be very close to $f(S^\circ_k)$. Actually, it works for any $S^*_b$ in a similar way according to Lemma \ref{lem7}.
\begin{lem}\label{lem7}
	Given $\delta'_1$, $\varepsilon'_1$, and $\theta_1$ as Eqn. (\ref{eq30}), if $\theta\geq\theta_1$, then we have $\Omega(S^*_b; \tilde{\mathcal{R}})/\theta>f(S^*_b)-\varepsilon'_1\cdot f(S^\circ_k)$ with at least $1-\delta'_1$ probability.
\end{lem}
\begin{proof}
	Similar to the proof of Lemma \ref{lem4}, by applying Eqn. (\ref{eq21}) in Lemma \ref{lem2}, we have $\Pr[\Omega(S^*_b; \tilde{\mathcal{R}})/\theta-f(S^*_b)\leq -\varepsilon'_1\cdot f(S^\circ_k)]=\Pr[\Omega(S^*_b; \tilde{\mathcal{R}})-\theta f(S^*_b)\leq(-\varepsilon'_1\cdot f(S^\circ_k)/f(S^*_b))\cdot\theta f(S^*_b)]\leq\exp(-{\varepsilon'_1}^2\cdot f(S^\circ_k)^2/(2f(S^*_b))\cdot\theta)\leq\exp(-\varepsilon'_1\cdot f(S^\circ_k)/2\cdot\theta_1)=\delta'_1$ by substituting $\theta_1$ with Eqn. (\ref{eq30}).
\end{proof}
\begin{lem}\label{lem8}
	Given $\delta'_1$ and $\varepsilon'_1$, if $f(S^*_b)\geq(1-(1-1/k)^b-\varepsilon_1)\cdot f(S^\circ_k)$, then we have
	\begin{equation}\label{eq33}
		\Omega(S^*_b\cup S^*_{k-b}; \tilde{\mathcal{R}})/\theta>(1-1/e-\varepsilon_1-\varepsilon'_1)\cdot f(S^\circ_k)
	\end{equation}
	with at least $1-2\delta'_1$ probability.
\end{lem}
\begin{proof}
	Based on Lemma \ref{lem6}, we have $\Omega(S^*_b\cup S^*_{k-b}; \tilde{\mathcal{R}})/\theta\geq(1-(1-1/k)^{k-b})\cdot\Omega(S^\circ_k; \tilde{\mathcal{R}})/\theta+(1-1/k)^{k-b}\cdot\Omega(S^*_b; \tilde{\mathcal{R}})/\theta\geq(1-(1-1/k)^{k-b})(1-\varepsilon'_1)\cdot f(S^\circ_k)+(1-1/k)^{k-b}\cdot(f(S^*_b)-\varepsilon'_1\cdot f(S^\circ_k))\geq(1-(1-1/k)^{k-b})(1-\varepsilon'_1)\cdot f(S^\circ_k)+(1-1/k)^{k-b}\cdot[(1-(1-1/k)^b-\varepsilon_1)\cdot f(S^\circ_k)-\varepsilon'_1\cdot f(S^\circ_k)]=(1-(1-1/k)^k-\varepsilon'_1-(1-1/k)^{k-b}\varepsilon_1)\cdot f(S^\circ_k)\geq(1-(1-1/k)^k-\varepsilon_1-\varepsilon'_1)\cdot f(S^\circ_k)$. Here, the second inequality is from Lemma \ref{lem4} and Lemma \ref{lem7}, where the $\Omega(S^\circ_k; \tilde{\mathcal{R}})/\theta> (1-\varepsilon'_1)\cdot f(S^\circ_k)$ holds with at least $1-\delta'_1$ probability and $\Omega(S^*_b; \tilde{\mathcal{R}})/\theta>f(S^*_b)-\varepsilon'_1\cdot f(S^\circ_k)$ holds with at least $1-\delta'_1$ probability. Thus, the Eqn. (\ref{eq33}) holds with at least $1-2\delta'_1$ probability by the union bound. Then, this lemma can be proven.
\end{proof}
\begin{lem}\label{lem9}
	Given $\delta'_2$, $\varepsilon'$ with $\varepsilon>\varepsilon'_1$, and
	\begin{equation}
		\theta_2=\frac{2(1-1/e)\cdot\left(\ln\binom{n-b}{k-b}+\ln\frac{1}{\delta'_2}\right)}{(\varepsilon'-\varepsilon'_1)\cdot f(S^\circ_k)},
	\end{equation}
	if Eqn. (\ref{eq33}) hold and $\theta\geq\theta_2$, then we have $f(S^*_b\cup S^*_{k-b})>(1-1/e-\varepsilon_1-\varepsilon')\cdot f(S^\circ_k)$ with at least $1-\delta'_2$ probability.
\end{lem}
\begin{proof}
	Let $S_{k-b}$ be any size-$(k-b)$ subset of $V\backslash S^*_b$. By applying Eqn. (\ref{eq22}) in Lemma \ref{lem2} and setting $\varepsilon'_2=\varepsilon'-\varepsilon'_1$, we have
	\begin{align}
		\Pr&\left[\Omega(S^*_b\cup S_{k-b}; \tilde{\mathcal{R}})/\theta-f(S^*_b\cup S_{k-b})\geq\varepsilon'_2\cdot f(S^\circ_k)\right]\nonumber\\
		&=\Pr\left[\Omega(S^*_b\cup S^{k-b}; \tilde{\mathcal{R}})-\theta f(S^*_b\cup S_{k-b})\right.\nonumber\\
		&\left.\quad\geq\frac{\varepsilon'_2\cdot f(S^\circ_k)}{f(S^*_b\cup S_{k-b})}\cdot \theta f(S^*_b\cup S_{k-b})\right]\nonumber\\
		&\leq\exp\left(-\frac{{\varepsilon'_2}^2\cdot f(S^\circ_k)^2}{2f(S^*_b\cup S_{k-b})+\frac{2}{3}\varepsilon'_2\cdot f(S^\circ_k)}\cdot\theta\right)\nonumber\\
		&\leq\exp\left(-\frac{{\varepsilon'_2}^2\cdot f(S^\circ_k)}{2(1-1/e-\varepsilon_1-\varepsilon')+\frac{2}{3}\varepsilon'_2}\cdot\theta\right)\nonumber\\
		&\leq\exp\left(-\frac{{\varepsilon'_2}^2\cdot f(S^\circ_k)}{2(1-1/e)}\cdot\theta_2\right)\nonumber\\
		&=\delta'_2/\binom{n-b}{k-b}.\nonumber
	\end{align}
	Thus, if $\theta\geq\theta_2$, the $S^*_{k-b}$ returned by the adapted MC-Greedy algorithm on $\tilde{\mathcal{R}}$ given the $S^*_b$ satisfies $\Omega(S^*_b\cup S^*_{k-b}; \tilde{\mathcal{R}})/\theta-f(S^*_b\cup S^*_{k-b})<\varepsilon'_2\cdot f(S^\circ_k)$ with at least $1-\delta'_2$ probability based on the union bound of at most $\binom{n-b}{k-b}$ size-$(k-b)$ node sets. Thus, if the Eqn. (\ref{eq33}) holds, we have
	\begin{align}
		f(S^*_b\cup S^*_{k-b})&>\Omega(S^*_b\cup S^*_{k-b}; \tilde{\mathcal{R}})/\theta-\varepsilon'_2\cdot f(S^\circ_k)\nonumber\\
		&>(1-1/e-\varepsilon_1-\varepsilon'_1-\varepsilon'_2)\cdot f(S^\circ_k)\nonumber\\
		&=(1-1/e-\varepsilon_1-\varepsilon')\cdot f(S^\circ_k)\nonumber
	\end{align}
	with at least $1-\delta'_2$ probability.
\end{proof}

Based on Lemma \ref{lem8} and Lemma \ref{lem9}, if $\theta\geq\max\{\theta_1,\theta_2\}$, when the first stage returns a good sentinel set $S^*_b$ with $f(S^*_b)\geq(1-(1-1/k)^b-\varepsilon_1)\cdot f(S^\circ_k)$, we have $f(S^*_b\cup S^*_{k-b})\geq(1-1/e-\varepsilon_1-\varepsilon')\cdot f(S^\circ_k)$ with at least $1-2\delta'_1-\delta'_2$ probability. By setting $\delta'_1=\delta'_2=\delta'/3$ and $\theta_1=\theta_2=\theta'$, we have $\theta'$ equals the Eqn. (\ref{eq28}) similar to the techniques in \cite{tang2015influence}. Theorem \ref{thm6} has been proven.

\subsection{Theoretical Result and Complexity}
In summary, based on Theorem \ref{thm5}, the sentinel set $S^*_b$ selected at the first stage satisfies $f(S^*_b)\geq(1-(1-1/k)^b-\varepsilon_1)\cdot f(S^\circ_b)$ with at least $1-\delta_1$ probability. When it works, based on Theorem \ref{thm6}, the remaining set $S^*_{k-b}$ selected at the second stage satisfies $f(S^*_b\cup S^*_{k-b})\geq(1-1/e-\varepsilon_1-\varepsilon_2)\cdot f(S^\circ_k)$ with at least $1-\delta_2$ probability. Shown as Algorithm \ref{a2}, by setting $\varepsilon_1=\varepsilon_2=\varepsilon$ and $\delta_1=\delta_2=\delta/2$, we have $f(S^*_b\cup S^*_{k-b})\geq(1-1/e-\varepsilon)\cdot f(S^\circ_k)$ with at least $1-\delta$ probability by the union bound.

For our G-HIST algorithm, the analysis of time complexity is very hard because the number $b$ at the first stage cannot be estimated. Thus, we only consider an extreme case, where there is no sentinel set selection stage, namely $b=0$. Then, the Algorithm \ref{a5} will directly select a size-$k$ seed set, whose process is similar to OPIM-C \cite{tang2018online}. Thus, when $\delta<1/2$, it generates an expected number of $O(k\ln n+\ln(1/\delta)/(\varepsilon^2\cdot f(S_k^\circ)))$ random G-RR sets. To generate a random G-RR set, the worst running time is less than $O(\sum_{q\in Q}r_q\cdot m)$. Thus, the worst time of generating random G-RR sets should be
\begin{equation}
	O\left(\frac{\sum_{q\in Q}r_q\cdot mn(k\ln n + \ln(1/\delta))}{\varepsilon^2\cdot k}\right).
\end{equation}
As the search time of MC-Greedy algorithm is shorter than the generation time, the total time complexity remains unchanged. Now, we can draw the main conclusion of this paper.
\begin{thm}[Main Theorem]
	The G-HIST shown as Algorithm \ref{a2} can be guaranteed to return a $(1-1/e-\varepsilon)$ approximate solution for the CC-DIM problem with at least $1-\delta$ probability and run in the $O(\sum_{q\in Q}r_q\cdot mn(k\ln n + \ln(1/\delta))/(\varepsilon^2\cdot k))$ worst expected time.
\end{thm}

\section{Experiments}
In this section, we conduct several experiments on different datasets to validate the effectiveness and efficiency of our G-HIST algorithm for the CC-DIM problem. All of our experiments are programmed by Python and run on a Mac machine. There are four datasets used in the experiments as follows. (1) NetScience \cite{nr}: A coauthorship network among scientists to publish papers about network science; (2) Wiki \cite{nr}: A who-votes-on-whom network coming from the collection of Wikipedia voting; (3) HetHEPT \cite{snapnets}: An academic collaboration relationship on high-energy physics area; and (4) Epinions \cite{snapnets}: A who-trust-whom online social network on Epinions.com, which is a general consumer review site. The statistics of these four datasets are shown in Table \ref{table1}. For the undirected graph, we replace each undirected edge with two reversed directed edges.

\begin{table}[h]
	\renewcommand{\arraystretch}{1.3}
	\caption{The datasets statistics $(K=10^3)$}
	\label{table1}
	\centering
	\begin{tabular}{|c|c|c|c|c|}
		\hline
		\bfseries Dataset & \bfseries n & \bfseries m & \bfseries Type & \bfseries Avg.Degree\\
		\hline
		NetScience & 0.4K & 1.01K & undirected & 5.00\\
		\hline
		Wiki & 1.0K & 3.15K & directed & 6.20\\
		\hline
		HetHEPT & 12.0K & 118.5K & undirected & 19.8\\
		\hline
		Epinions & 75.9K & 508.8K & directed & 13.4\\
		\hline
	\end{tabular}
\end{table}

\subsection{Experimental Settings}
For the IC model, we use the Weighted Cascade (WC) \cite{guo2019targeted} \cite{guo2020multi} \cite{guo2021continuous} to set the diffusion probability of each edge. The probability $p_{uv}$ for each edge $(u,v)\in E$ is $1/|N^-(v)$. As for the parameters in G-HIST algorithm, we set $\varepsilon=0.1$ and $\delta=0.1$. We conduct $1000$ Monte Carlo simulations to estimate the objective function given a seed set. Each point in our result is the average over 3 times running.

Because our CC-DIM is a composite community-aware problem, there are multiple community structures in a shared social network. In the objective function shown as Eqn. (\ref{of}), we give the $\lambda=0.7$. Then, we consider three cases of different number of community structures as follows. (1) Case 1: One community structure, denoted by $Q_1=\{q_1\}$; (2) Case 2: Two community structure, denoted by $Q_2=\{q_1,q_2\}$; and (3) Case 3: Three community structure, denoted by $Q_3=\{q_1,q_2,q_3\}$. Here, we have $r_{q_1}=3$, $r_{q_2}=4$, and $r_{q_3}=5$, where the graph will be partitioned into three communities under the metric $q_1$, four communities under the metric $q_2$, and five communities under the metric $q_2$. Besides, for each of the above cases, we will give two different settings as follows. Under the parameter setting 1, we have $\{w_{q_1}=1\}$ for $Q_1$, $\{w_{q_1}=0.4,w_{q_2}=0.6\}$ for $Q_2$, and $\{w_{q_1}=0.3,w_{q_2}=0.3,w_{q_3}=0.4\}$ for $Q_3$; $\{a_{q_1,1}=0.4,a_{q_1,2}=1,a_{q_1,3}=1.6\}$, $\{a_{q_2,1}=0.4,a_{q_2,2}=0.8,a_{q_2,3}=1.2,a_{q_2,4}=1.6\}$, and $\{a_{q_3,1}=0.2,a_{q_3,2}=0.6,a_{q_3,3}=1,a_{q_3,4}=1.4,a_{q_3,5}=1.8\}$. Under the parameter setting 2, we have $\{w_{q_1}=1\}$ for $Q_1$, $\{w_{q_1}=0.1,w_{q_2}=0.9\}$ for $Q_2$, and $\{w_{q_1}=0.1,w_{q_2}=0.1,w_{q_3}=0.8\}$ for $Q_3$; $\{a_{q_1,1}=0.1,a_{q_1,2}=0.1,a_{q_1,3}=2.8\}$, $\{a_{q_2,1}=0.1,a_{q_2,2}=0.1,a_{q_2,3}=0.8,a_{q_2,4}=3\}$, and $\{a_{q_3,1}=0.1,a_{q_3,2}=0.1,a_{q_3,3}=0.1,a_{q_3,4}=1.7,a_{q_3,5}=3\}$. We will explain them in the later analysis.

How to get a community partition is flexible. We can not only divide the community according to user's attributes given by datasets, but also we can use existing algorithms \cite{girvan2002community} \cite{chen2014community} \cite{karrer2011stochastic} \cite{shi2000normalized} to divide the community.

Next, we introduce some typical baselines, which will be used for comparison with our G-HIST algorithm.
\begin{itemize}
	\item G-HIST: It is given by Algorithm \ref{a2}.
	\item G-HIST-no-Sentinel: It directly selects a size-$k$ seed set without the stage of sentinel set selection by invoking Algorithm \ref{a5} like RemainingSet $(G,Q,k,\emptyset,\varepsilon, \varepsilon, \delta)$.
	\item G-IMM: It uses the IMM \cite{tang2015influence} to maximize our objective function by setting $\varepsilon=0.1$ and $\delta=0.1$.
	\item Greedy: It adopts greedy hill-climbing algorithm to select the node with maximum marginal gain in each iteration through Monte Carlo simulations.
	\item Greedy-IM: It uses greedy hill-climbing algorithm to solve the IM problem through Monte Carlo simulations.
	\item IMM \cite{tang2015influence}: A classic sampling based method of the IM problem by setting $\varepsilon=0.1$ and $\delta=0.1$.
	\item MaxDegree: It selects the node with maximum out degree in each iteration.
	\item Random: It randomly select a size-$k$ seed set.
\end{itemize}

\begin{figure}[!t]
	\centering
	\subfigure[NetScence, $Q_1$, Parameter 1]{
		\includegraphics[width=0.48\linewidth]{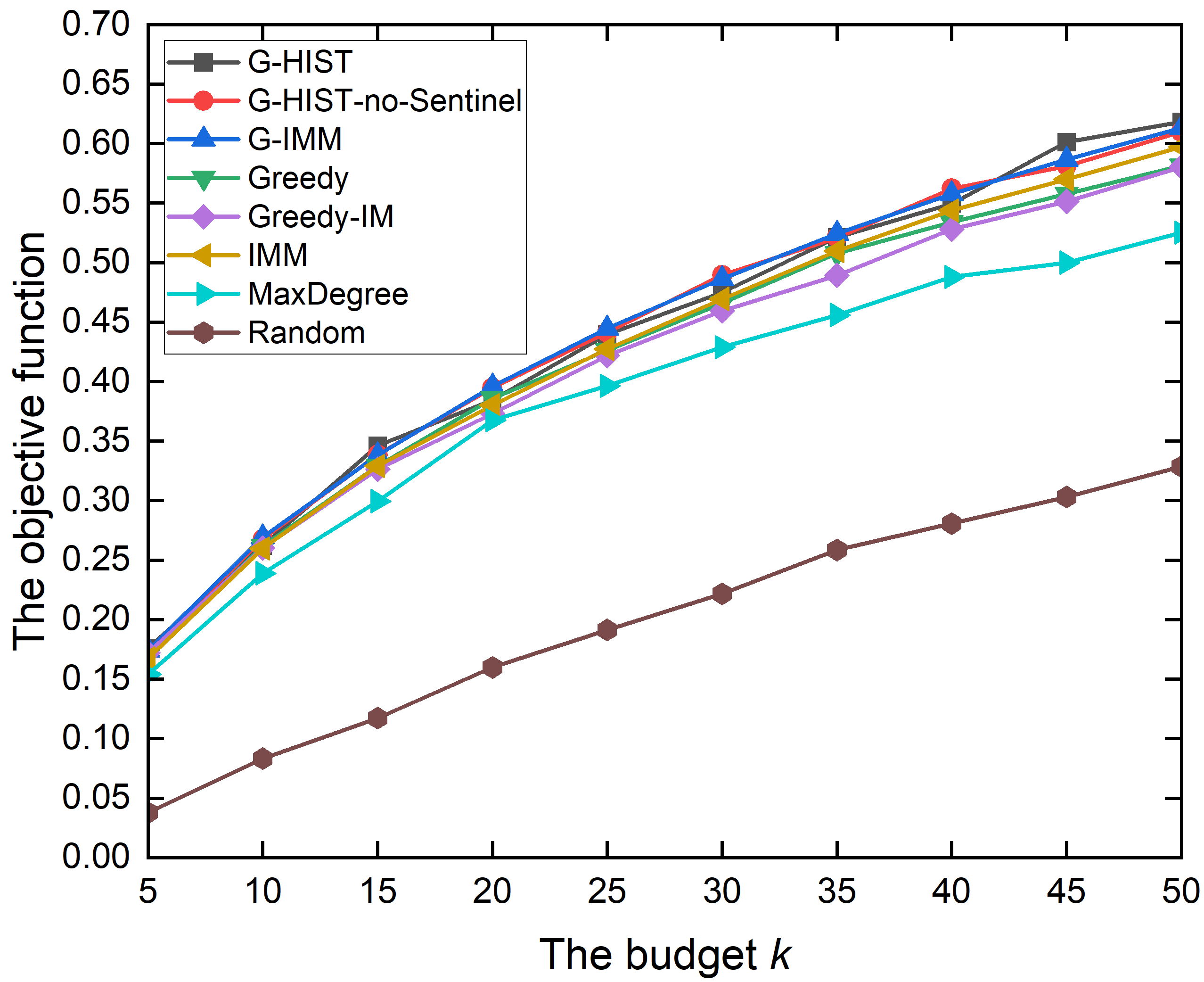}
	}%
	\subfigure[NetScence, $Q_2$, Parameter 1]{
		\includegraphics[width=0.48\linewidth]{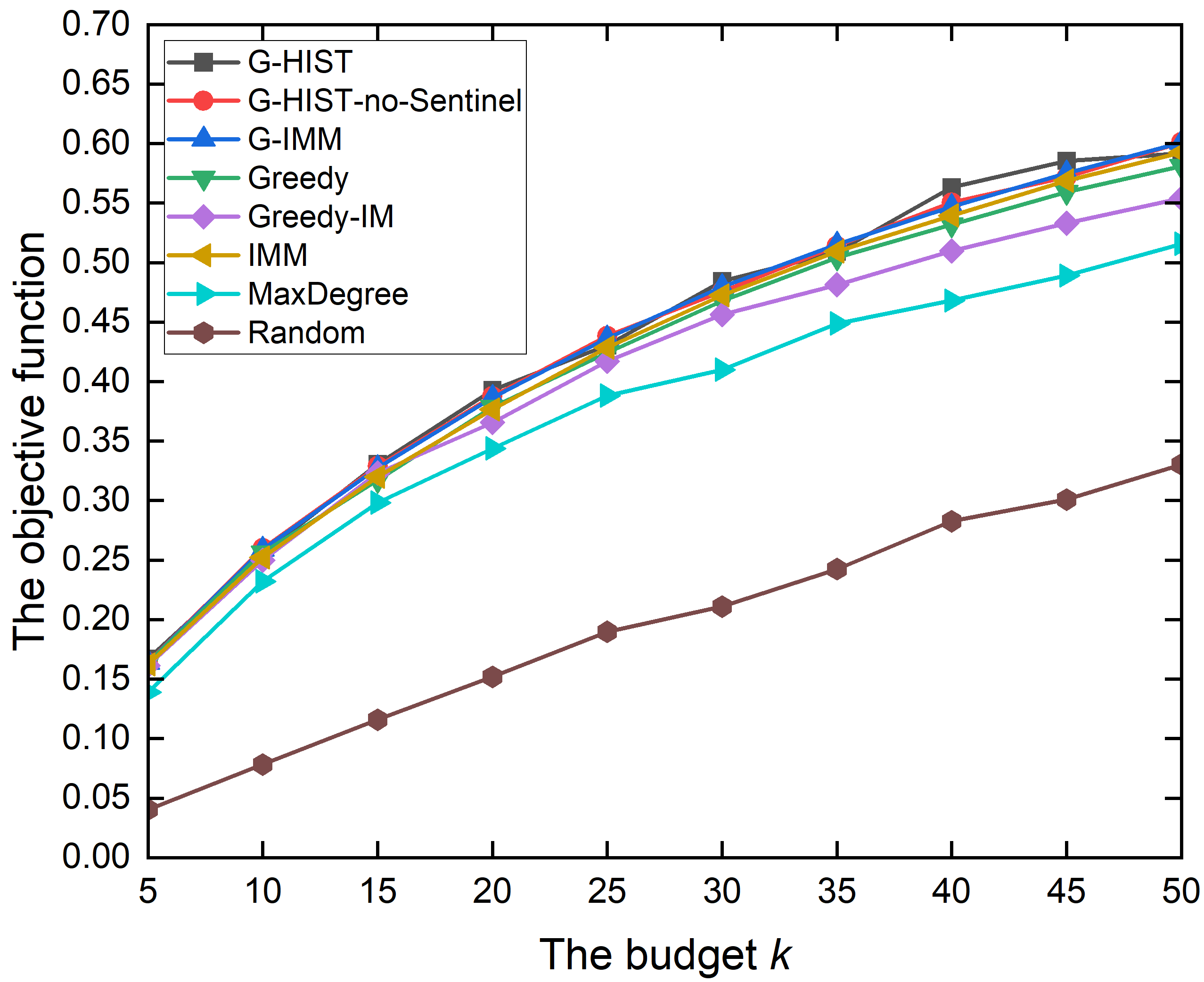}
	}%
	\centering
	
	\subfigure[NetScience, $Q_3$, Parameter 1]{
		\includegraphics[width=0.48\linewidth]{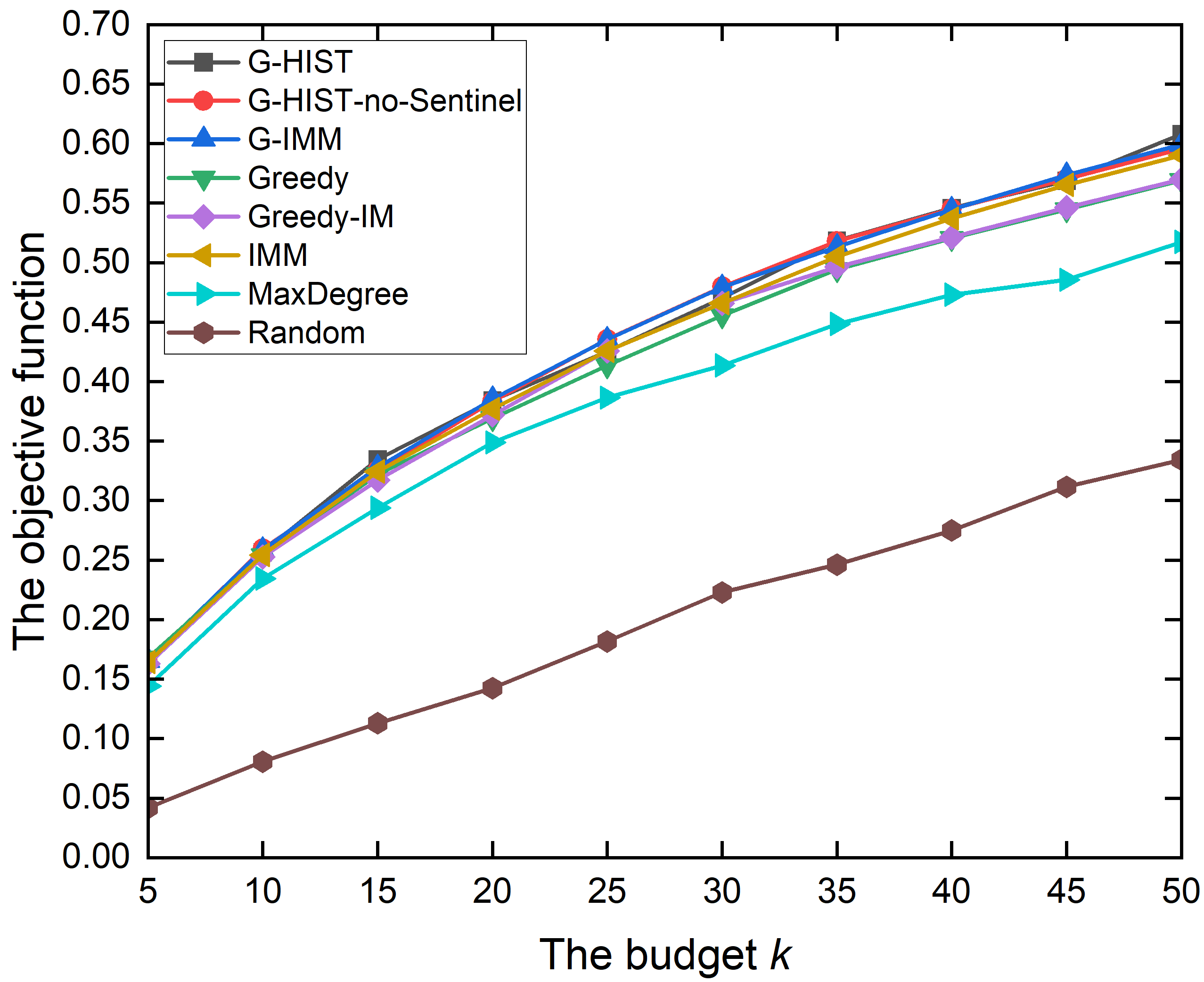}
	}%
	\subfigure[NetScience, $Q_1$, Parameter 2]{
		\includegraphics[width=0.48\linewidth]{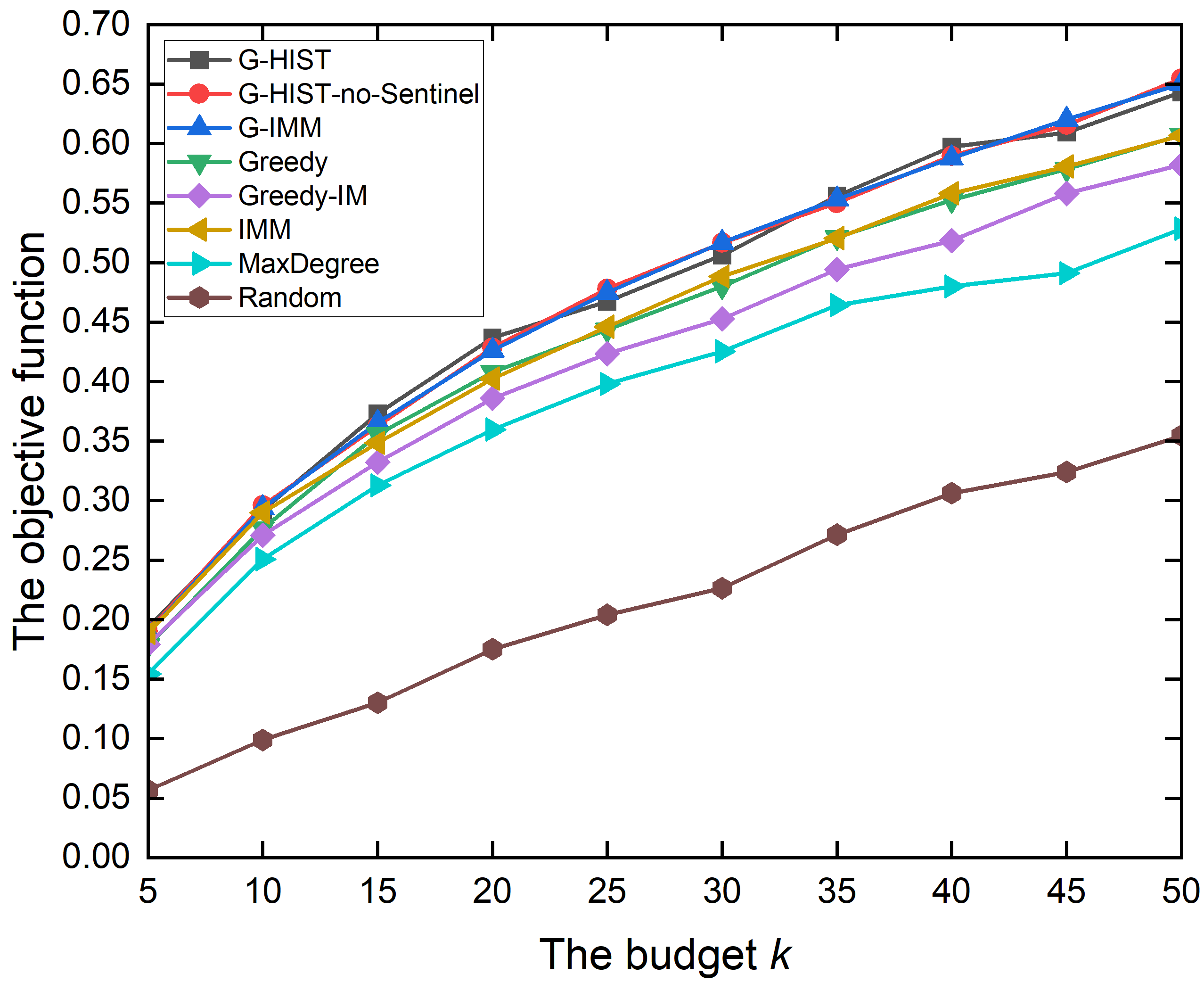}
	}%
	\centering
	
	\subfigure[NetScence, $Q_2$, Parameter 2]{
		\includegraphics[width=0.48\linewidth]{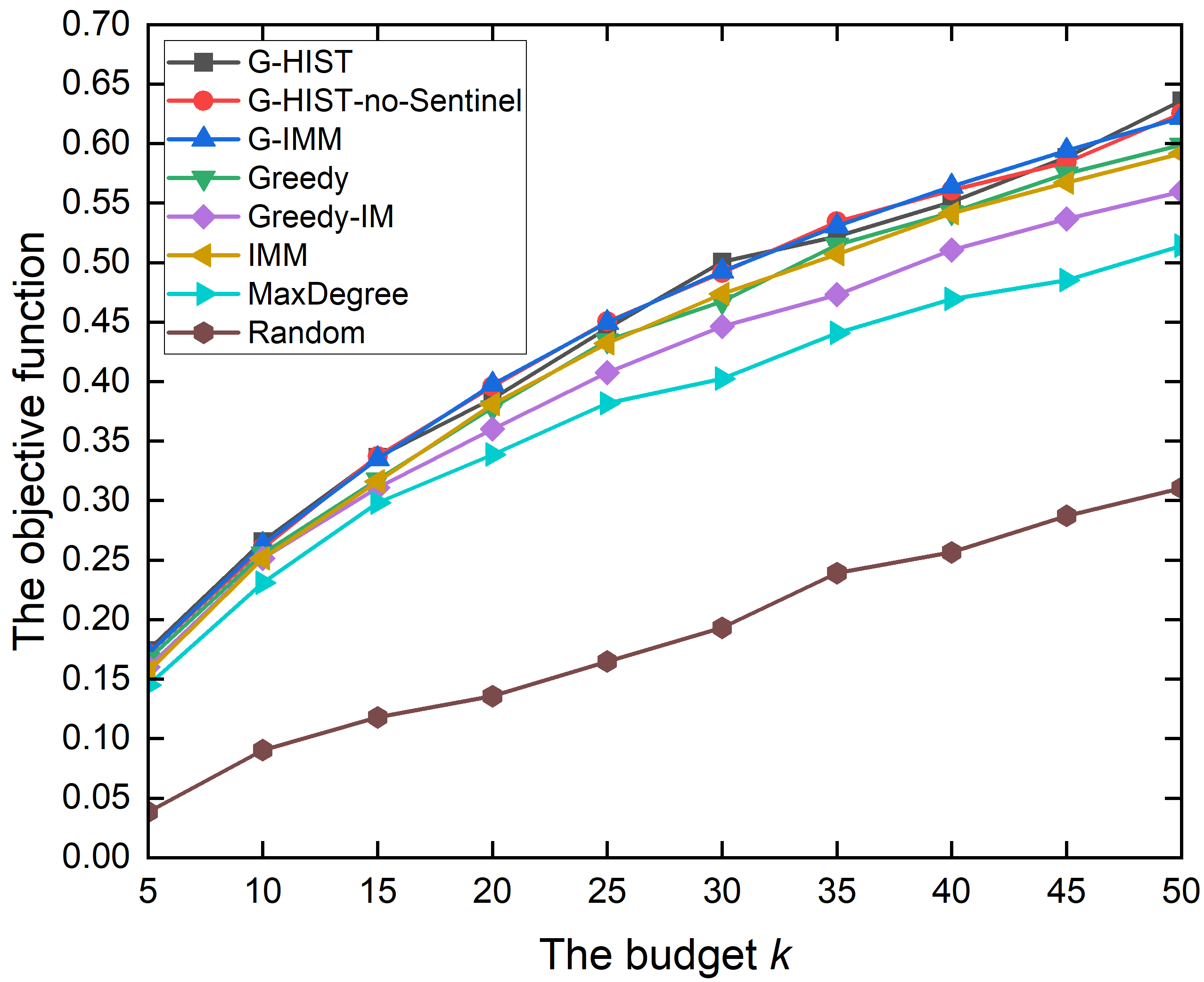}
	}%
	\subfigure[NetScence, $Q_3$, Parameter 2]{
		\includegraphics[width=0.48\linewidth]{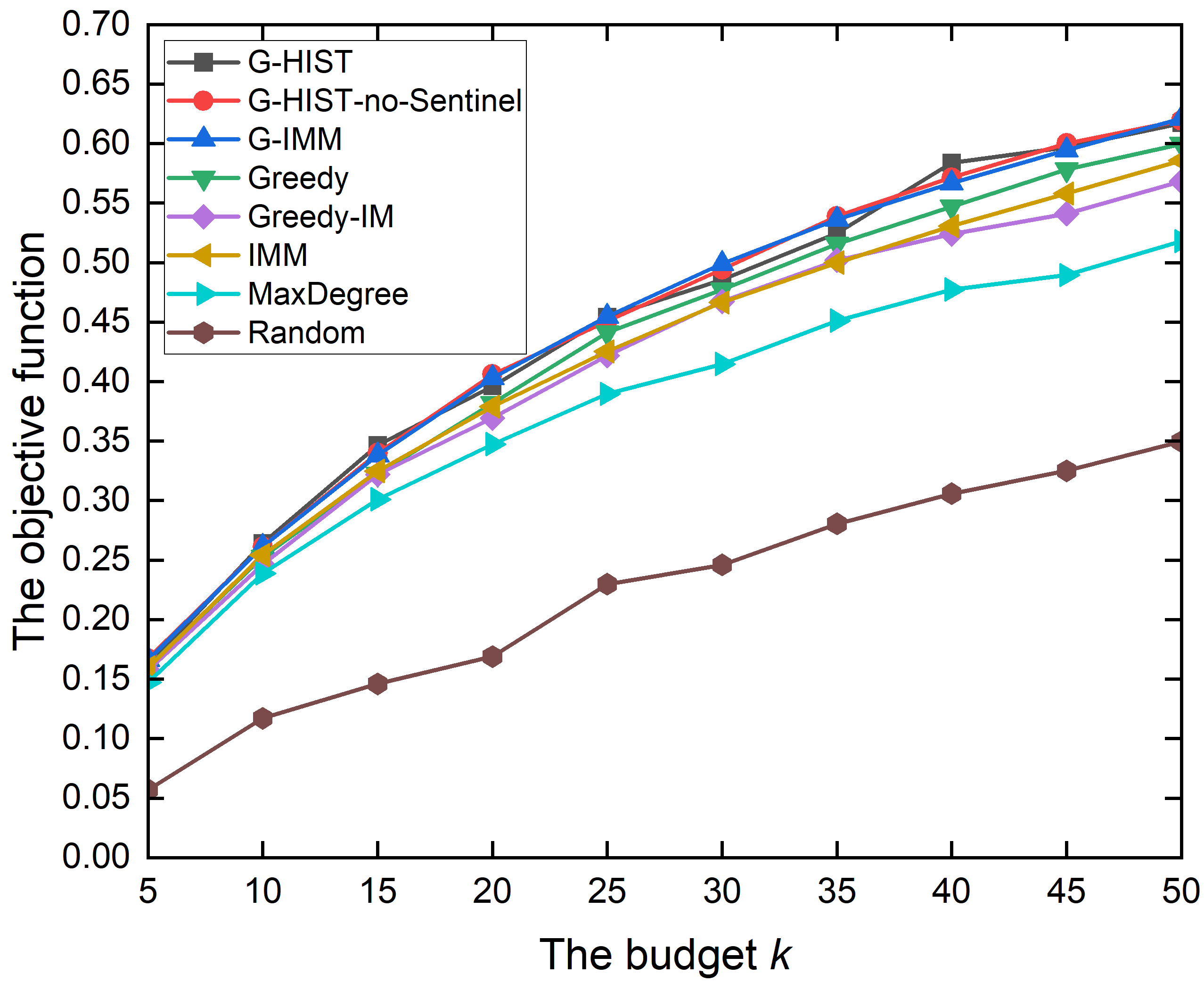}
	}%
	\centering
	\caption{The performance comparison achieved by the NetScience dataset under three community structures and two parameter settings.}
	\label{fig1}
\end{figure}

\begin{figure}[!t]
	\centering
	\subfigure[Wiki, $Q_1$, Parameter 1]{
		\includegraphics[width=0.48\linewidth]{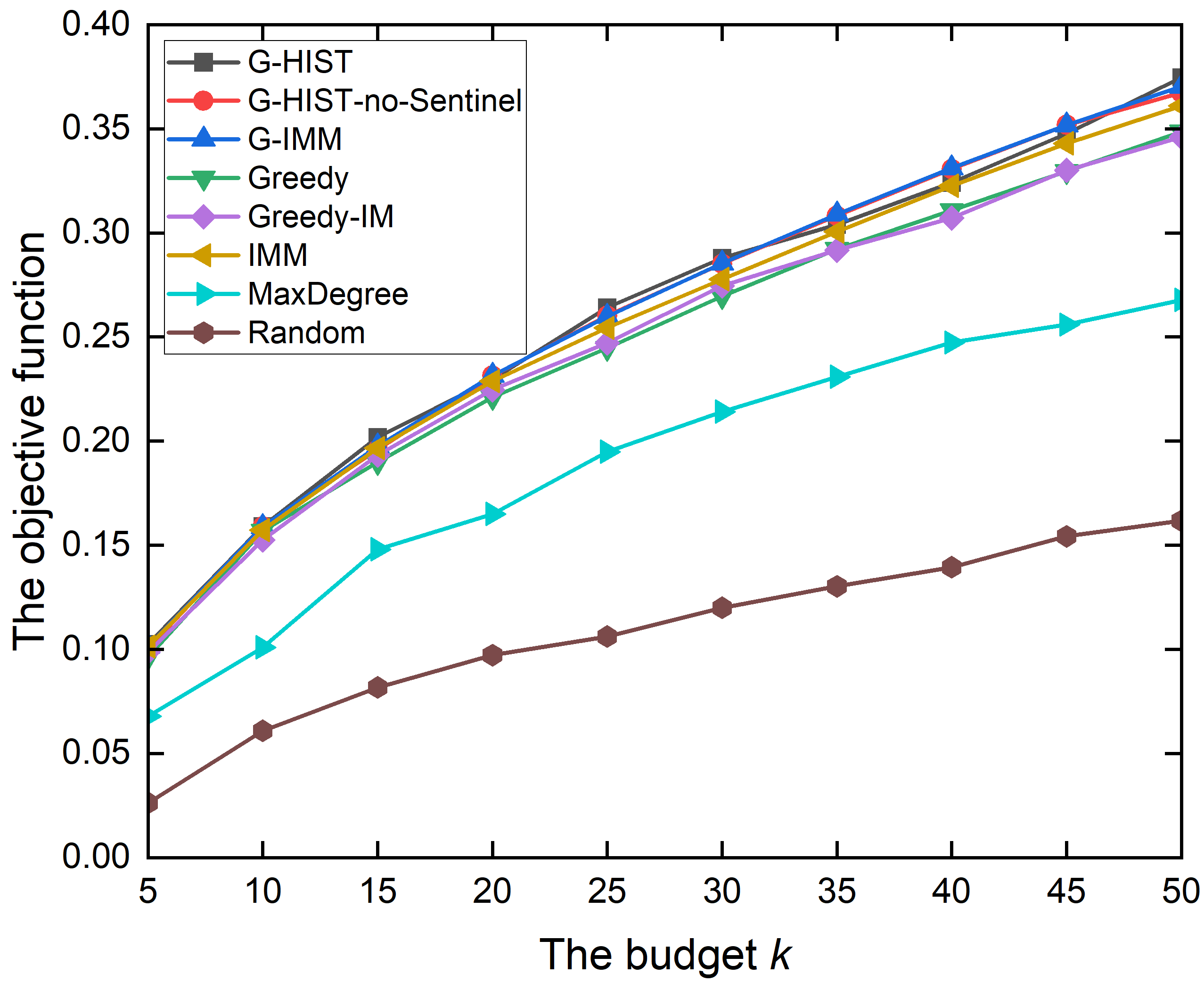}
	}%
	\subfigure[Wiki, $Q_2$, Parameter 1]{
		\includegraphics[width=0.48\linewidth]{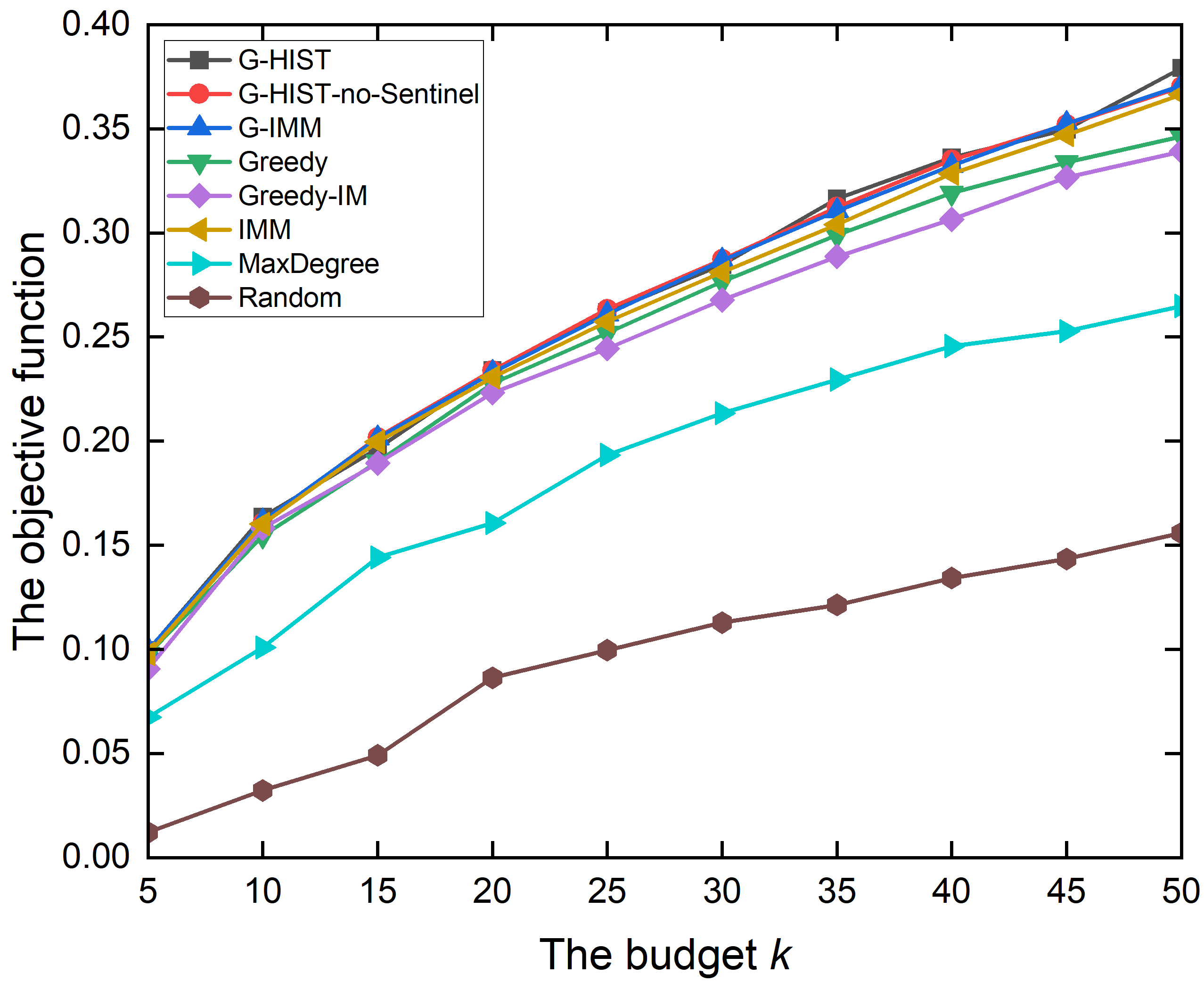}
	}%
	\centering
	
	\subfigure[Wiki, $Q_3$, Parameter 1]{
		\includegraphics[width=0.48\linewidth]{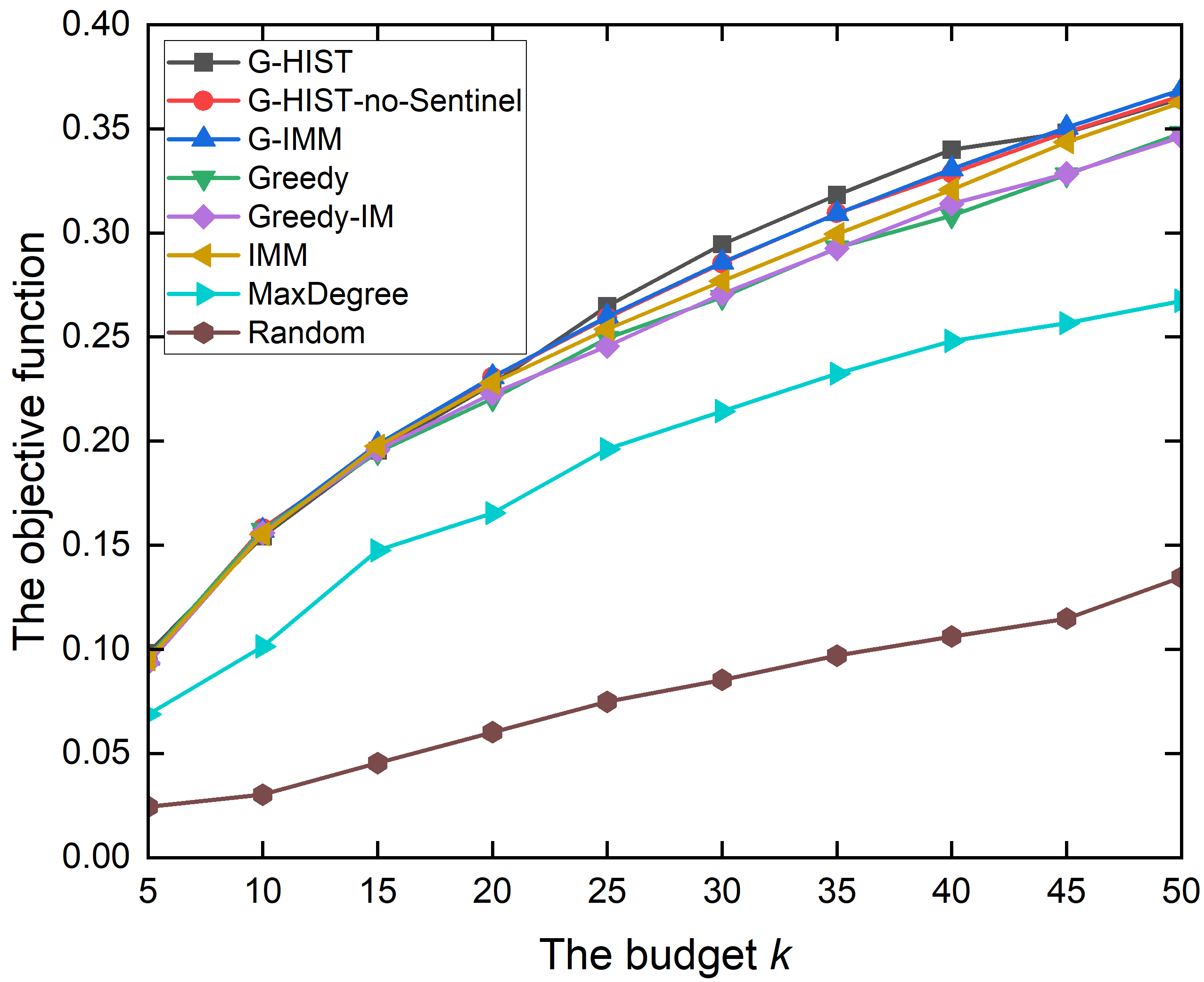}
	}%
	\subfigure[Wiki, $Q_1$, Parameter 2]{
		\includegraphics[width=0.48\linewidth]{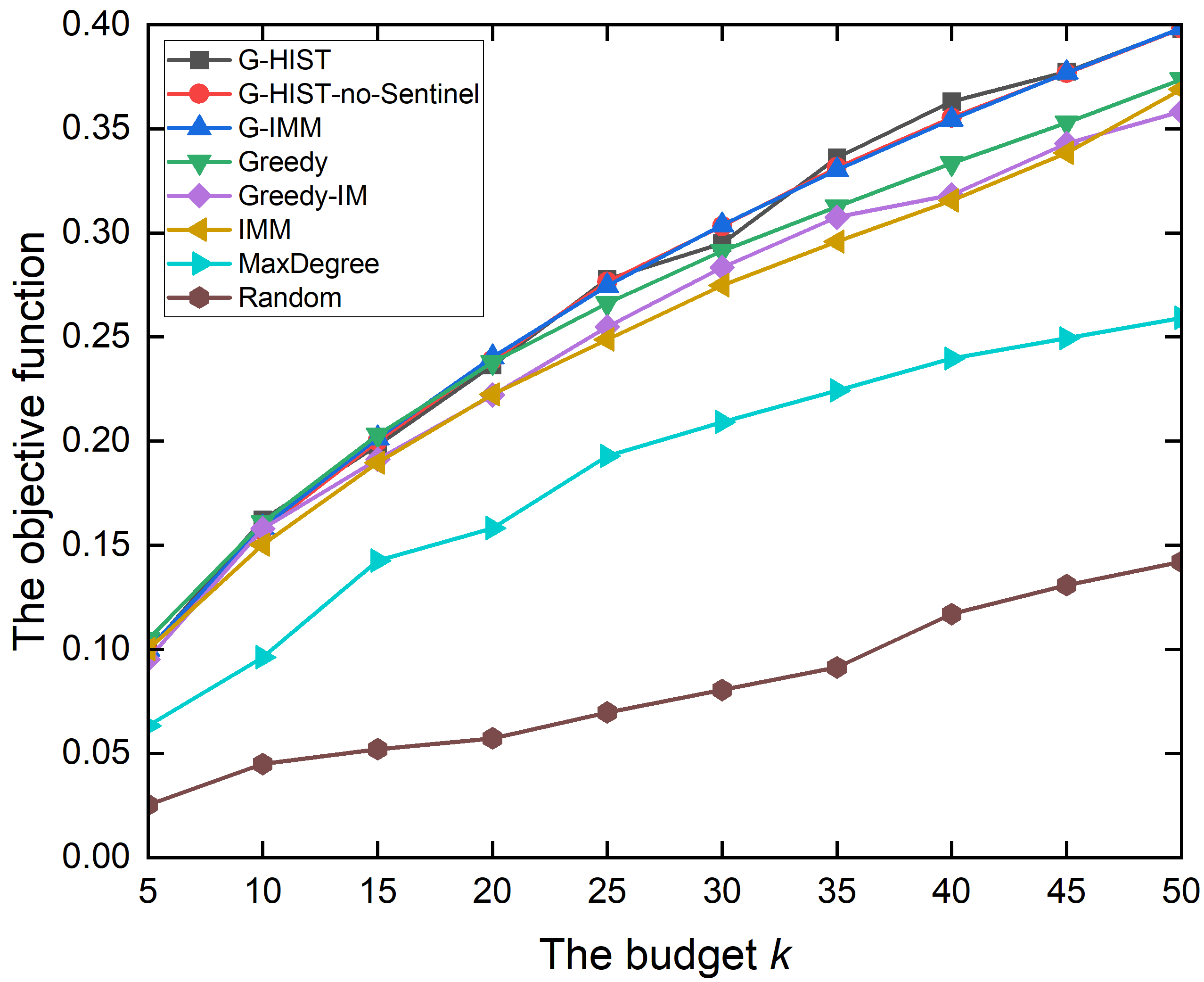}
	}%
	\centering
	
	\subfigure[Wiki, $Q_2$, Parameter 2]{
		\includegraphics[width=0.48\linewidth]{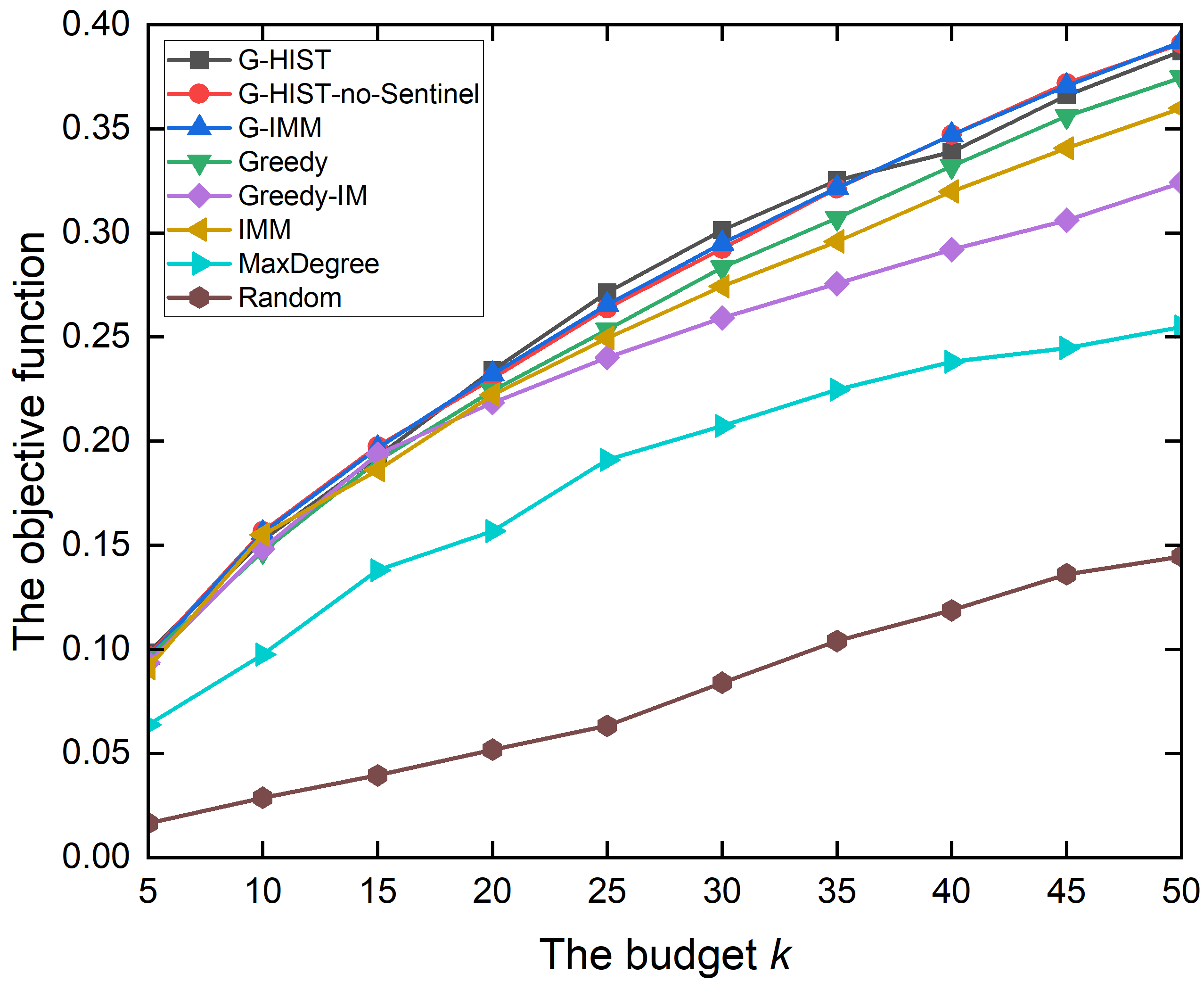}
	}%
	\subfigure[Wiki, $Q_3$, Parameter 2]{
		\includegraphics[width=0.48\linewidth]{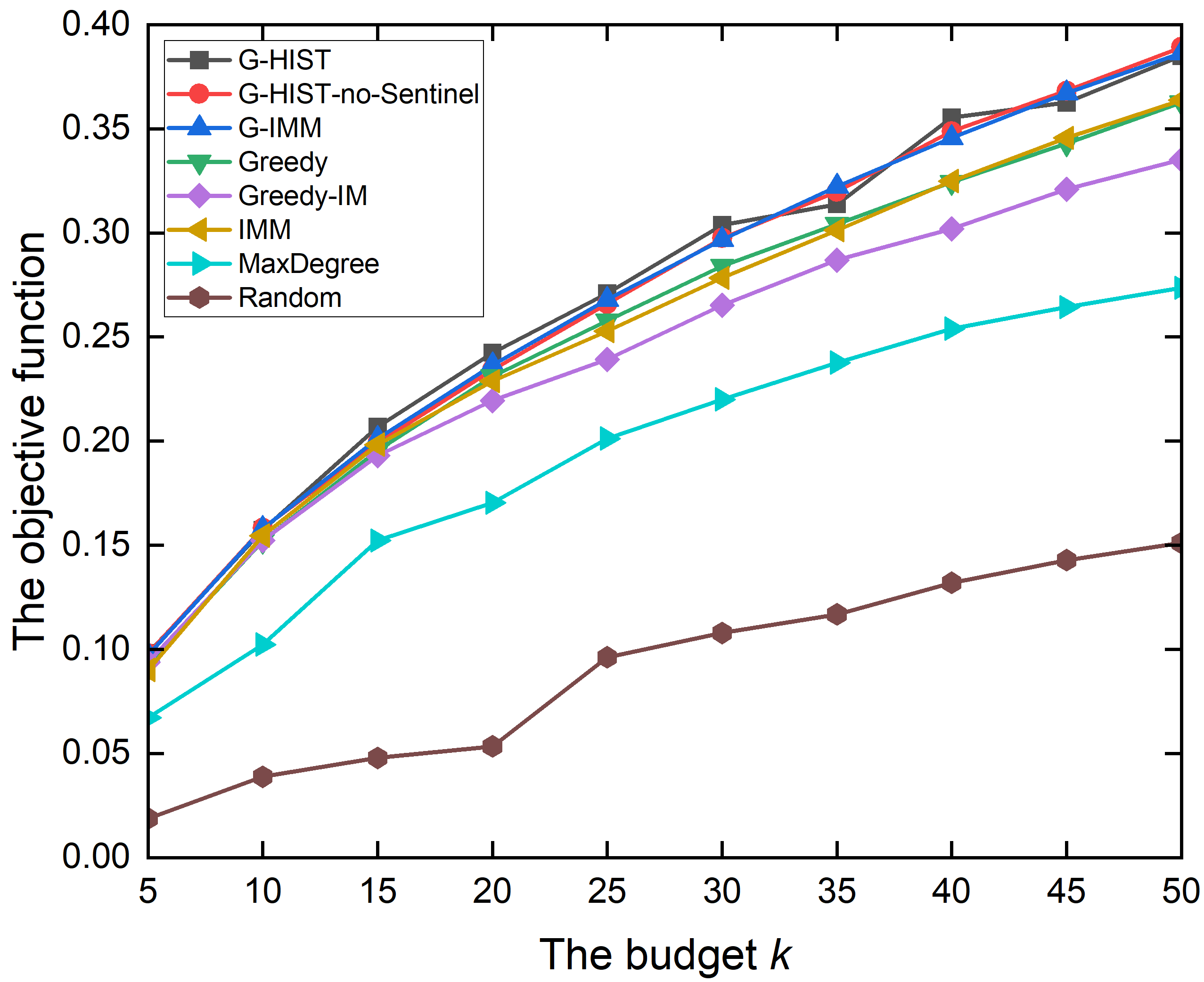}
	}%
	\centering
	\caption{The performance comparison achieved by the Wiki dataset under three community structures and two parameter settings.}
	\label{fig2}
\end{figure}

\begin{figure}[!t]
	\centering
	\subfigure[HetHEPT, $Q_1$, Parameter 2]{
		\includegraphics[width=0.48\linewidth]{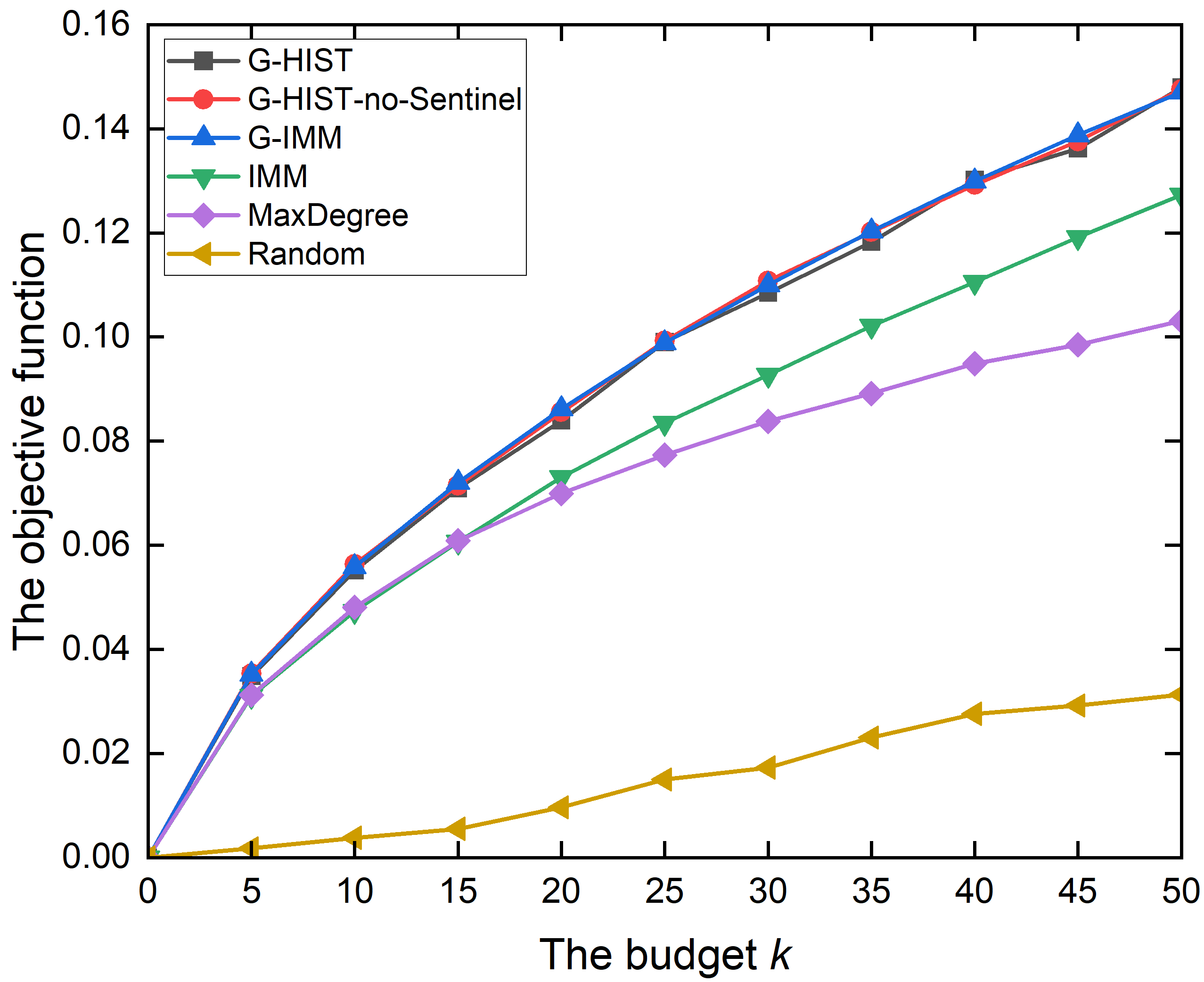}
	}%
	\subfigure[HetHEPT, $Q_2$, Parameter 2]{
		\includegraphics[width=0.48\linewidth]{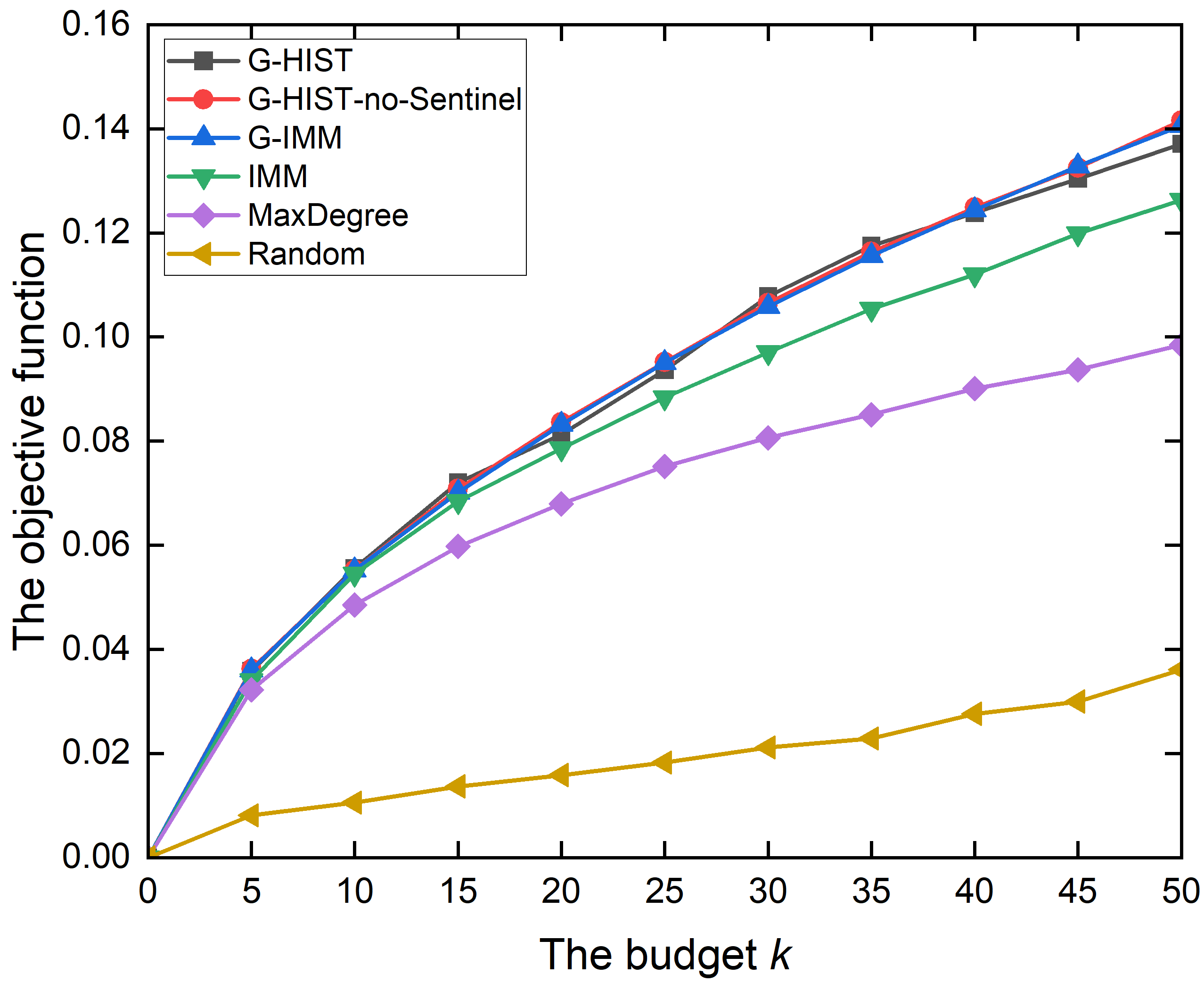}
	}%
	\centering
	
	\subfigure[HetHEPT, $Q_3$, Parameter 2]{
		\includegraphics[width=0.48\linewidth]{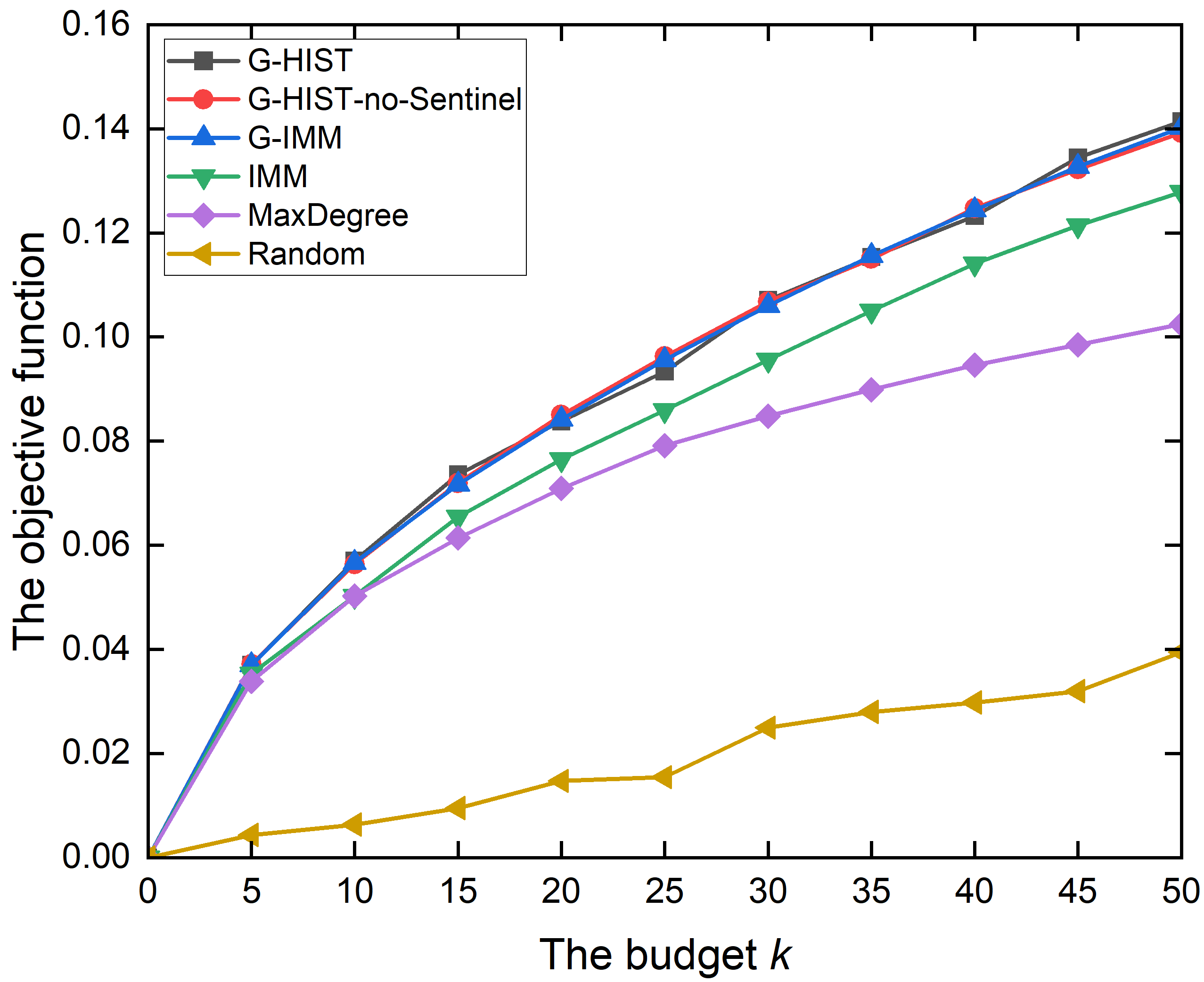}
	}%
	\subfigure[Epinions, $Q_1$, Parameter 2]{
		\includegraphics[width=0.48\linewidth]{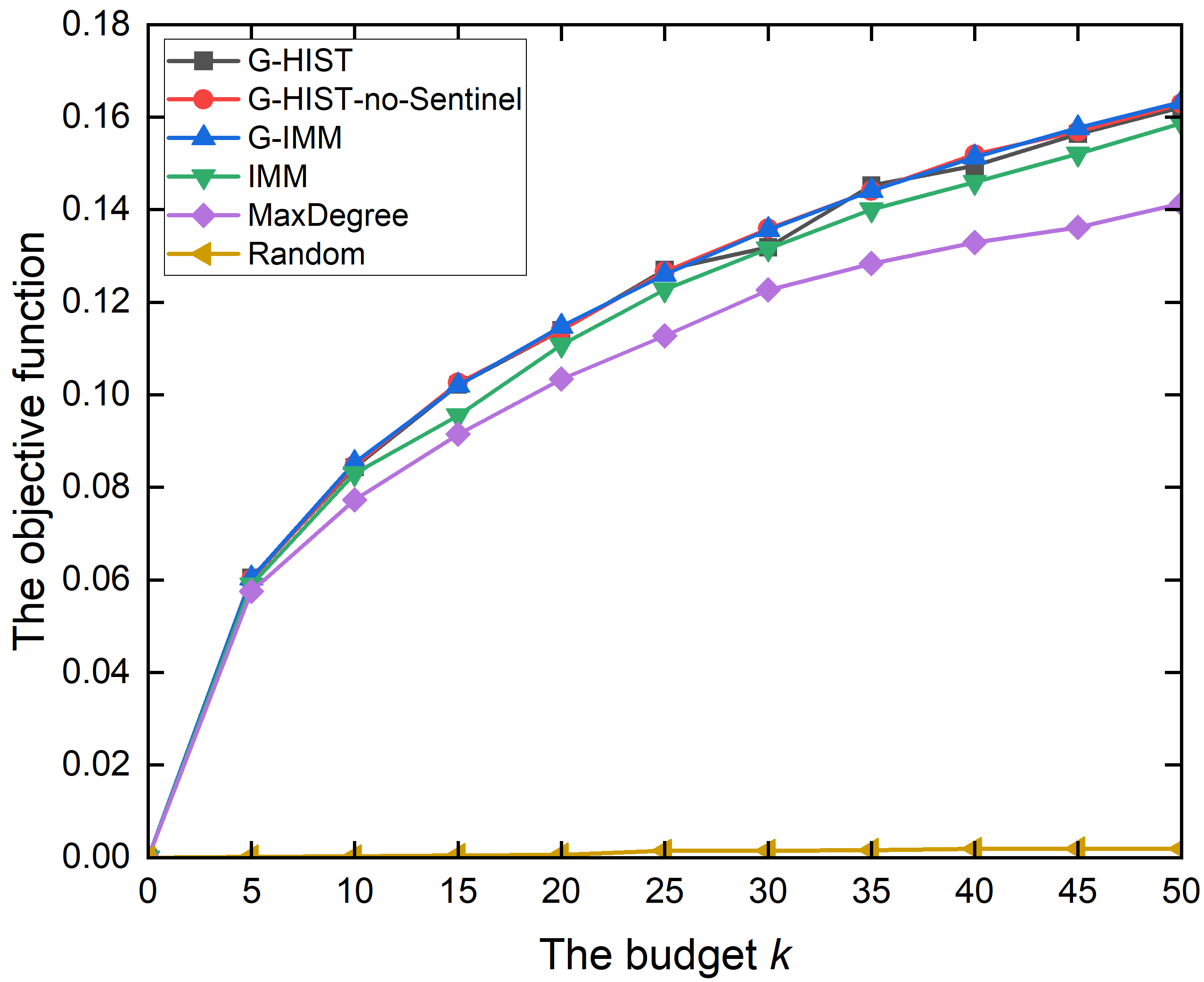}
	}%
	\centering
	
	\subfigure[Epinions, $Q_2$, Parameter 2]{
		\includegraphics[width=0.48\linewidth]{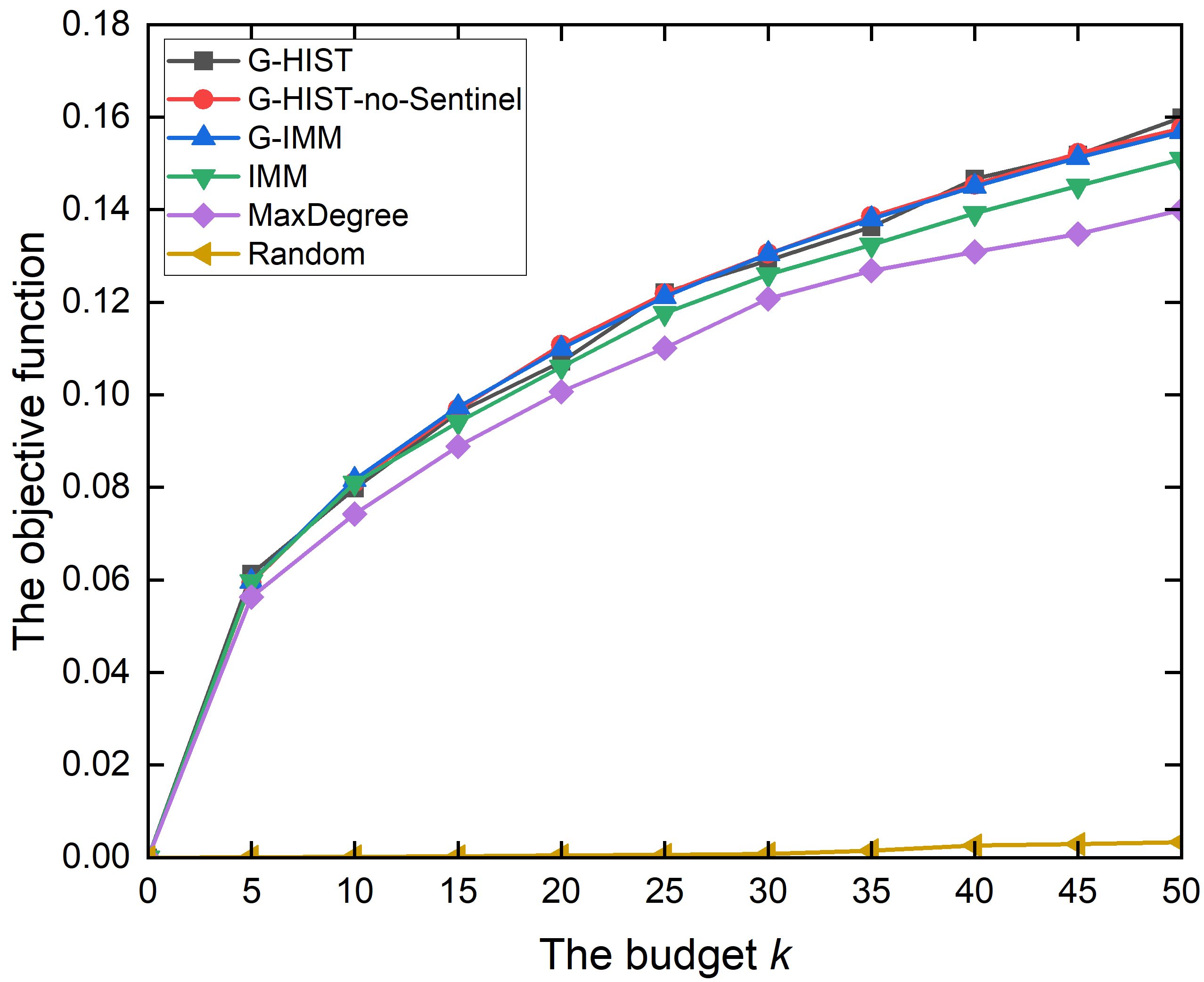}
	}%
	\subfigure[Epinions, $Q_3$, Parameter 2]{
		\includegraphics[width=0.48\linewidth]{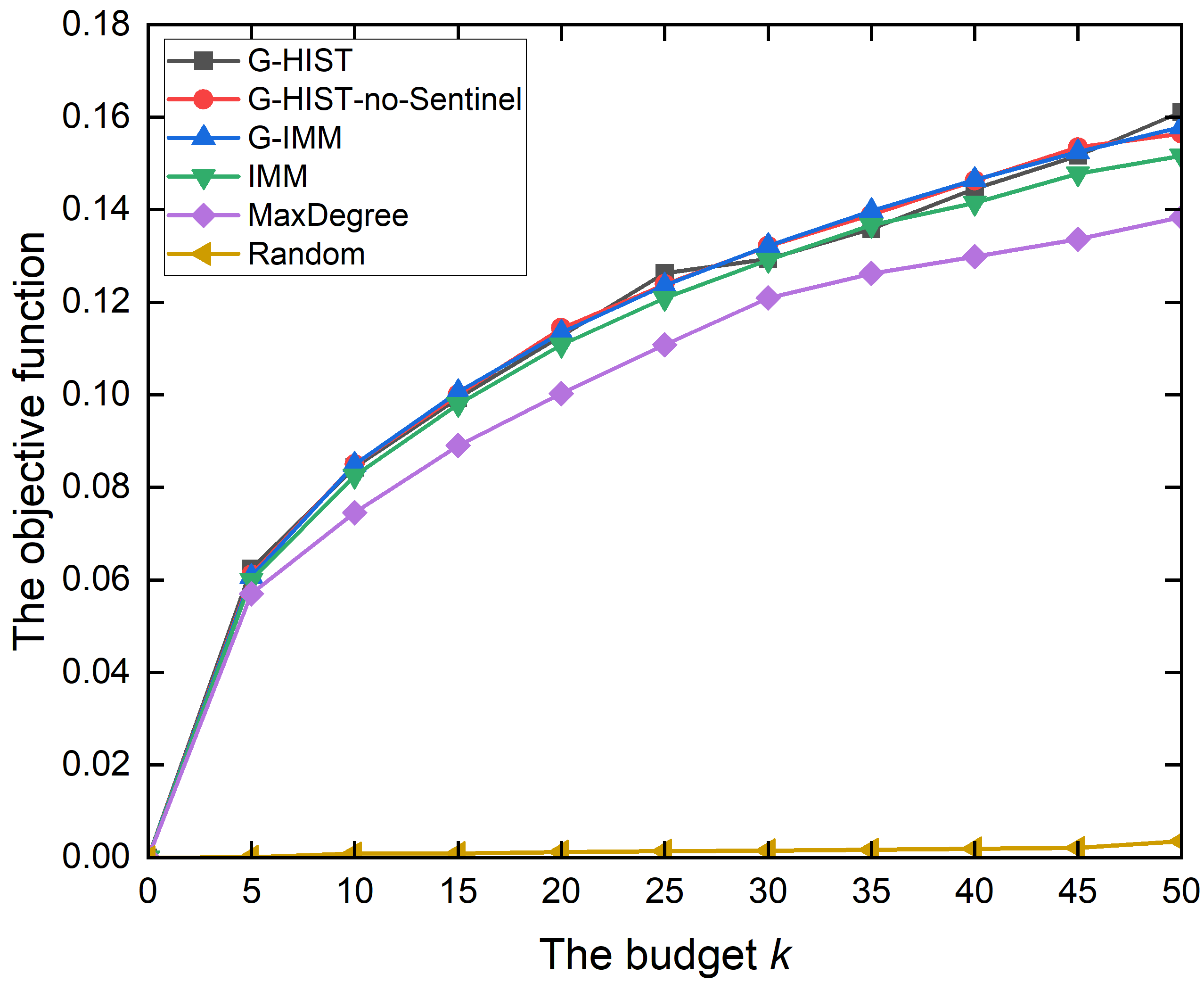}
	}%
	\centering
	\caption{The performance comparison achieved by the HetHEPT and Epinions dataset under three community structures and parameter setting 2.}
	\label{fig3}
\end{figure}

\subsection{Experimental Results}
\subsubsection{Performance} Fig. \ref{fig1} and Fig. \ref{fig2} draw the performance comparison achieved by all kinds of algorithms with different settings under the NetScience dataset and Wiki dataset. In this part, we only consider these two smaller datasets since the Greedy and Greedy-IM algorithm are implemented by Monte Carlo simulations, which cannot be carried out in an acceptable running time for a large network. Here, we make the following observations. First, no matter which case and parameter setting we use, the performances obtained by G-HIST, G-HIST-no-Sentinel, and G-IMM algorithm are very close, and they are obviously better than other baselines. This means that our G-HIST algorithm and pre-selected sentinel set will not significantly reduce performance, even if the performance fluctuates slightly more than other baselines. Thus, we think that the G-HIST, G-HIST-no-Sentinel, and G-IMM are consistent in performance. Second, through comparing G-HIST with Greedy and comparing IMM with Greedy-IM, we can see that the performance of G-HIST (resp. IMM) is slightly better than that of Greedy (resp. Greedy-IM), which means that the methods based on sampling are slightly better than the corresponding methods based on Monte Carlo simulations. In fact, they should be roughly equal. This may be due to the insufficient number of simulations, which leads to inaccurate estimation of the objective function.

Third, through comparing G-HIST (resp. Greedy) with IMM (resp. Greedy-IM), we find that their performance is roughly similar under the parameter setting 1, where the performance of G-HIST is just a little better than the IMM. However, under the parameter setting 2, the performance of G-HIST is significantly superior to the IMM, and the performance gap between G-HIST and IMM will increase as the budget $k$ increases. This is because the weight and coefficient distributions are more uneven under the parameter setting 2, which leads to obvious bias. Here, the IMM (Greedy-IM) algorithm that ignores the requirement of diversity will result in obvious reduction of the objective function, but this effect is not significant under the parameter setting 1 where the weight and coefficient distributions are uniform. Fourth, the performance of G-HIST has large advantage over the MaxDegree and Random algorithm, and the performance gap between G-HIST and MaxDegree will increase as the budget $k$ increases.

\subsubsection{Running time}First, the computational cost of simulation-based methods, Greedy and Greedy-IM, is much higher than the corresponding sampling-based methods, G-IMM and IMM. Thus, we will not adopt simulation-based methods in the later experiment. Second, similar to OPIM-C \cite{tang2018online}, we find a similar trend between G-HIST-no-Sentinel and G-IMM. Under the equivalent setting, the running time of G-HIST-no-Sentinel is about 20\% to 30\% of G-IMM. Third, through comparing G-HIST with G-HIST-no-Sentinel, the average size of random G-RR sets at the second stage of G-HIST can be reduced by nearly one order of magnitude, and the running time can also be significantly improved.

\subsubsection{Scalability}Fig. \ref{fig3} draws the performance comparison achieved by sampling-based algorithms under the HetHEPT and Epinions dataset. In the large networks, our G-HIST has the same advantages in performance and running time as before, which further verifies its effectiveness.

\section{Conclusion}
To tackle the multiplicity of diversity in real social applications, in this paper, we first propose the Composite Community-aware Diversified IM (CC-DIM) problem, which is totally different from the traditional IM problem and Diversified IM problem. Even though its objective function is monotone and submodular, it is extremely hard to compute. Thus, we create a novel sampling method based on Generalized Reverse Reachable (G-RR) set to effectively estimate the objective function, and design a two-stage G-HIST algorithm to further improve the memory consumption and time efficiency by significantly reducing the average size of random G-RR sets. According to our theoretical analysis, the G-HIST returns a $(1-1/e-\varepsilon)$ approximate solution with at least $(1-\delta)$ probability in an acceptable running time. Finally, our experimental results verify our theories and demonstrate the effectiveness and correctness of our proposed algorithm over other the-state-of-art baselines.

However, in order to ensure the rigor of theoretical guarantee, our sampling and algorithm design are conservative. There is still a lot of room for improvement in the future.

\section*{Acknowledgment}

This work was supported in part by the CCF-Huawei Populus Grove Fund under Grant No. CCF-HuaweiLK2022004, the National Natural Science Foundation of China (NSFC) under Grant No. 62202055, the Start-up Fund from Beijing Normal University under Grant No. 310432104, the Start-up Fund from BNU-HKBU United International College under Grant No. UICR0700018-22, and the Project of Young Innovative Talents of Guangdong Education Department under Grant No. 2022KQNCX102.

\ifCLASSOPTIONcaptionsoff
  \newpage
\fi

\bibliographystyle{IEEEtran}
\bibliography{references}

\begin{thebibliography}{10}
\providecommand{\url}[1]{#1}
\csname url@samestyle\endcsname
\providecommand{\newblock}{\relax}
\providecommand{\bibinfo}[2]{#2}
\providecommand{\BIBentrySTDinterwordspacing}{\spaceskip=0pt\relax}
\providecommand{\BIBentryALTinterwordstretchfactor}{4}
\providecommand{\BIBentryALTinterwordspacing}{\spaceskip=\fontdimen2\font plus
\BIBentryALTinterwordstretchfactor\fontdimen3\font minus
  \fontdimen4\font\relax}
\providecommand{\BIBforeignlanguage}[2]{{%
\expandafter\ifx\csname l@#1\endcsname\relax
\typeout{** WARNING: IEEEtran.bst: No hyphenation pattern has been}%
\typeout{** loaded for the language `#1'. Using the pattern for}%
\typeout{** the default language instead.}%
\else
\language=\csname l@#1\endcsname
\fi
#2}}
\providecommand{\BIBdecl}{\relax}
\BIBdecl

\bibitem{kempe2003maximizing}
D.~Kempe, J.~Kleinberg, and {\'E}.~Tardos, ``Maximizing the spread of influence
  through a social network,'' in \emph{Proceedings of the ninth ACM SIGKDD
  international conference on Knowledge discovery and data mining}, 2003, pp.
  137--146.

\bibitem{chen2015online}
S.~Chen, J.~Fan, G.~Li, J.~Feng, K.-l. Tan, and J.~Tang, ``Online topic-aware
  influence maximization,'' \emph{Proceedings of the VLDB Endowment}, vol.~8,
  no.~6, pp. 666--677, 2015.

\bibitem{tian2020deep}
S.~Tian, S.~Mo, L.~Wang, and Z.~Peng, ``Deep reinforcement learning-based
  approach to tackle topic-aware influence maximization,'' \emph{Data Science
  and Engineering}, vol.~5, no.~1, pp. 1--11, 2020.

\bibitem{tong2020time}
G.~Tong, R.~Wang, Z.~Dong, and X.~Li, ``Time-constrained adaptive influence
  maximization,'' \emph{IEEE Transactions on Computational Social Systems},
  vol.~8, no.~1, pp. 33--44, 2020.

\bibitem{guo2021adaptive}
J.~Guo and W.~Wu, ``Adaptive influence maximization: if influential node
  unwilling to be the seed,'' \emph{ACM Transactions on Knowledge Discovery
  from Data (TKDD)}, vol.~15, no.~5, pp. 84:1--84:23, 2021.

\bibitem{li2014efficient}
G.~Li, S.~Chen, J.~Feng, K.-l. Tan, and W.-s. Li, ``Efficient location-aware
  influence maximization,'' in \emph{Proceedings of the 2014 ACM SIGMOD
  international conference on Management of data}, 2014, pp. 87--98.

\bibitem{chen2020efficient}
X.~Chen, Y.~Zhao, G.~Liu, R.~Sun, X.~Zhou, and K.~Zheng, ``Efficient
  similarity-aware influence maximization in geo-social network,'' \emph{IEEE
  Transactions on Knowledge and Data Engineering}, 2020.

\bibitem{guo2019targeted}
J.~Guo, Y.~Li, and W.~Wu, ``Targeted protection maximization in social
  networks,'' \emph{IEEE Transactions on Network Science and Engineering},
  vol.~7, no.~3, pp. 1645--1655, 2019.

\bibitem{cai2020target}
T.~Cai, J.~Li, A.~S. Mian, T.~Sellis, J.~X. Yu \emph{et~al.}, ``Target-aware
  holistic influence maximization in spatial social networks,'' \emph{IEEE
  Transactions on Knowledge and Data Engineering}, vol.~34, no.~4, pp.
  1993--2007, 2020.

\bibitem{yin2019social}
H.~Yin, Q.~Wang, K.~Zheng, Z.~Li, J.~Yang, and X.~Zhou, ``Social
  influence-based group representation learning for group recommendation,'' in
  \emph{2019 IEEE 35th International Conference on Data Engineering
  (ICDE)}.\hskip 1em plus 0.5em minus 0.4em\relax IEEE, 2019, pp. 566--577.

\bibitem{zhang2020personalized}
Y.~Zhang, G.~Liu, A.~Liu, Y.~Zhang, Z.~Li, X.~Zhang, and Q.~Li, ``Personalized
  geographical influence modeling for poi recommendation,'' \emph{IEEE
  Intelligent Systems}, vol.~35, no.~5, pp. 18--27, 2020.

\bibitem{tang2014diversified}
F.~Tang, Q.~Liu, H.~Zhu, E.~Chen, and F.~Zhu, ``Diversified social influence
  maximization,'' in \emph{2014 IEEE/ACM International Conference on Advances
  in Social Networks Analysis and Mining (ASONAM 2014)}.\hskip 1em plus 0.5em
  minus 0.4em\relax IEEE, 2014, pp. 455--459.

\bibitem{zhang2019diversifying}
Y.~Zhang, ``Diversifying seeds and audience in social influence maximization,''
  in \emph{Proceedings of the 2019 IEEE/ACM International Conference on
  Advances in Social Networks Analysis and Mining}, 2019, pp. 89--94.

\bibitem{li2020community}
J.~Li, T.~Cai, K.~Deng, X.~Wang, T.~Sellis, and F.~Xia, ``Community-diversified
  influence maximization in social networks,'' \emph{Information Systems},
  vol.~92, p. 101522, 2020.

\bibitem{wang2021efficient}
C.~Wang, Q.~Shi, W.~Xian, Y.~Feng, and C.~Chen, ``Efficient diversified
  influence maximization with adaptive policies,'' \emph{Knowledge-Based
  Systems}, vol. 213, p. 106692, 2021.

\bibitem{zhang2021grain}
W.~Zhang, Z.~Yang, Y.~Wang, Y.~Shen, Y.~Li, L.~Wang, and B.~Cui, ``Grain:
  improving data efficiency of graph neural networks via diversified in fluence
  maximization,'' \emph{Proceedings of the VLDB Endowment}, vol.~14, no.~11,
  pp. 2473--2482, 2021.

\bibitem{chen2010scalable}
W.~Chen, C.~Wang, and Y.~Wang, ``Scalable influence maximization for prevalent
  viral marketing in large-scale social networks,'' in \emph{Proceedings of the
  16th ACM SIGKDD international conference on Knowledge discovery and data
  mining}, 2010, pp. 1029--1038.

\bibitem{chen2010scalable2}
W.~Chen, Y.~Yuan, and L.~Zhang, ``Scalable influence maximization in social
  networks under the linear threshold model,'' in \emph{2010 IEEE international
  conference on data mining}.\hskip 1em plus 0.5em minus 0.4em\relax IEEE,
  2010, pp. 88--97.

\bibitem{borgs2014maximizing}
C.~Borgs, M.~Brautbar, J.~Chayes, and B.~Lucier, ``Maximizing social influence
  in nearly optimal time,'' in \emph{Proceedings of the twenty-fifth annual
  ACM-SIAM symposium on Discrete algorithms}.\hskip 1em plus 0.5em minus
  0.4em\relax SIAM, 2014, pp. 946--957.

\bibitem{tang2014influence}
Y.~Tang, X.~Xiao, and Y.~Shi, ``Influence maximization: Near-optimal time
  complexity meets practical efficiency,'' in \emph{Proceedings of the 2014 ACM
  SIGMOD international conference on Management of data}, 2014, pp. 75--86.

\bibitem{tang2015influence}
Y.~Tang, Y.~Shi, and X.~Xiao, ``Influence maximization in near-linear time: A
  martingale approach,'' in \emph{Proceedings of the 2015 ACM SIGMOD
  international conference on management of data}, 2015, pp. 1539--1554.

\bibitem{nguyen2016stop}
H.~T. Nguyen, M.~T. Thai, and T.~N. Dinh, ``Stop-and-stare: Optimal sampling
  algorithms for viral marketing in billion-scale networks,'' in
  \emph{Proceedings of the 2016 international conference on management of
  data}, 2016, pp. 695--710.

\bibitem{tang2018online}
J.~Tang, X.~Tang, X.~Xiao, and J.~Yuan, ``Online processing algorithms for
  influence maximization,'' in \emph{Proceedings of the 2018 International
  Conference on Management of Data}, 2018, pp. 991--1005.

\bibitem{guo2020influence}
Q.~Guo, S.~Wang, Z.~Wei, and M.~Chen, ``Influence maximization revisited:
  Efficient reverse reachable set generation with bound tightened,'' in
  \emph{Proceedings of the 2020 ACM SIGMOD International Conference on
  Management of Data}, 2020, pp. 2167--2181.

\bibitem{guo2022influence}
Q.~Guo, S.~Wang, Z.~Wei, W.~Lin, and J.~Tang, ``Influence maximization
  revisited: Efficient sampling with bound tightened,'' \emph{ACM Transactions
  on Database Systems}, vol.~47, no.~3, pp. 12:1--12:45, 2022.

\bibitem{girvan2002community}
M.~Girvan and M.~E. Newman, ``Community structure in social and biological
  networks,'' \emph{Proceedings of the national academy of sciences}, vol.~99,
  no.~12, pp. 7821--7826, 2002.

\bibitem{chen2014community}
M.~Chen, K.~Kuzmin, and B.~K. Szymanski, ``Community detection via maximization
  of modularity and its variants,'' \emph{IEEE Transactions on Computational
  Social Systems}, vol.~1, no.~1, pp. 46--65, 2014.

\bibitem{karrer2011stochastic}
B.~Karrer and M.~E. Newman, ``Stochastic blockmodels and community structure in
  networks,'' \emph{Physical review E}, vol.~83, no.~1, p. 016107, 2011.

\bibitem{shi2000normalized}
J.~Shi and J.~Malik, ``Normalized cuts and image segmentation,'' \emph{IEEE
  Transactions on pattern analysis and machine intelligence}, vol.~22, no.~8,
  pp. 888--905, 2000.

\bibitem{nemhauser1978analysis}
G.~L. Nemhauser, L.~A. Wolsey, and M.~L. Fisher, ``An analysis of
  approximations for maximizing submodular set functions—i,''
  \emph{Mathematical programming}, vol.~14, no.~1, pp. 265--294, 1978.

\bibitem{lin2011class}
H.~Lin and J.~Bilmes, ``A class of submodular functions for document
  summarization,'' in \emph{Proceedings of the 49th annual meeting of the
  association for computational linguistics: human language technologies},
  2011, pp. 510--520.

\bibitem{nr}
\BIBentryALTinterwordspacing
R.~A. Rossi and N.~K. Ahmed, ``The network data repository with interactive
  graph analytics and visualization,'' in \emph{AAAI}, 2015. [Online].
  Available: \url{http://networkrepository.com}
\BIBentrySTDinterwordspacing

\bibitem{snapnets}
J.~Leskovec and A.~Krevl, ``{SNAP Datasets}: {Stanford} large network dataset
  collection,'' \url{http://snap.stanford.edu/data}, Jun. 2014.

\bibitem{guo2020multi}
J.~Guo, T.~Chen, and W.~Wu, ``A multi-feature diffusion model: Rumor blocking
  in social networks,'' \emph{IEEE/ACM Transactions on Networking}, vol.~29,
  no.~1, pp. 386--397, 2020.

\bibitem{guo2021continuous}
J.~Guo and W.~Wu, ``Continuous profit maximization: a study of unconstrained
  dr-submodular maximization,'' \emph{IEEE Transactions on Computational Social
  Systems}, vol.~8, no.~3, pp. 768--779, 2021.

\end{thebibliography}

\begin{IEEEbiography}[{\includegraphics[width=1in,height=1.25in,clip,keepaspectratio]{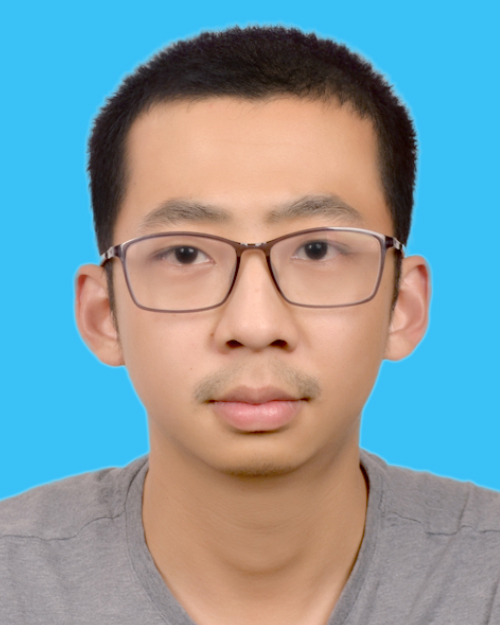}}]{Jianxiong Guo}
	received his Ph.D. degree from the Department of Computer Science, University of Texas at Dallas, Richardson, TX, USA, in 2021, and his B.E. degree from the School of Chemistry and Chemical Engineering, South China University of Technology, Guangzhou, China, in 2015. He is currently an Assistant Professor with the Advanced Institute of Natural Sciences, Beijing Normal University, and also with the Guangdong Key Lab of AI and Multi-Modal Data Processing, BNU-HKBU United International College, Zhuhai, China. He is a member of IEEE/ACM/CCF. He has published more than 40 peer-reviewed papers and been the reviewer for many famous international journals/conferences. His research interests include social networks, wireless sensor networks, combinatorial optimization, and machine learning.
\end{IEEEbiography} 

\begin{IEEEbiography}[{\includegraphics[width=1in,height=1.25in,clip,keepaspectratio]{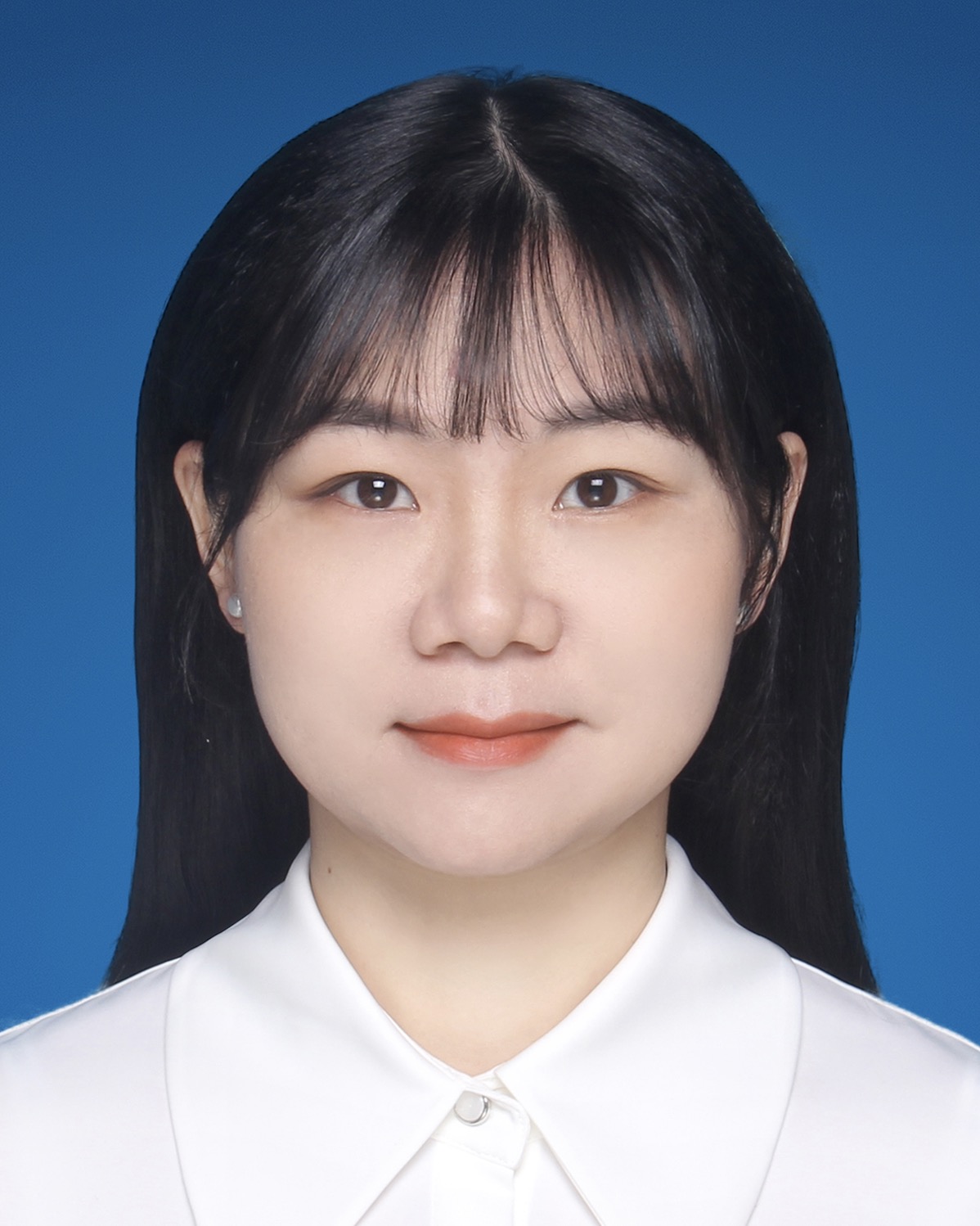}}]{Qiufen Ni}
is currently an Assistant Professor with the School of Computers, Guangdong University of Technology, Guangzhou, Guangdong 510006, China. She received the Ph.D. degree from School of Computers, Wuhan University in Dec. 2020. During Oct. 2018 to Oct. 2020, she was also a joint Ph.D. student in the Department of Computer Science, University of Texas at Dallas, Richardson, TX, USA. Her research interests include social networks, theoretical approximation algorithm design and analysis, and optimization problems in wireless networks. She serves in China Computer Federation as technical committee member in the Theoretical Computer Science branch committee.
\end{IEEEbiography} 

\begin{IEEEbiography}[{\includegraphics[width=1in,height=1.25in,clip,keepaspectratio]{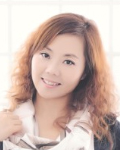}}]{Weili Wu}
	received the Ph.D. and M.S. degrees from the Department of Computer Science, University of Minnesota, Minneapolis, MN, USA, in 2002 and 1998, respectively. She is currently a Full Professor with the Department of Computer Science, The University of Texas at Dallas, Richardson, TX, USA. Her research mainly deals with the general research area of data communication and data management. Her research focuses on the design and analysis of algorithms for optimization problems that occur in wireless networking environments and various database systems.
\end{IEEEbiography}

\begin{IEEEbiography}[{\includegraphics[width=1in,height=1.25in,clip,keepaspectratio]{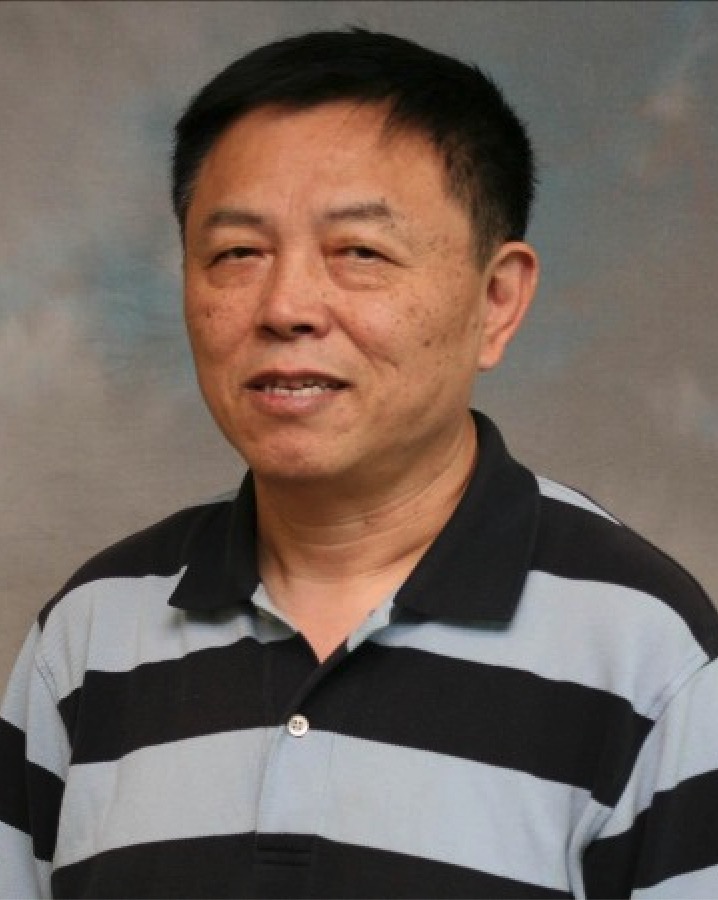}}]{Ding-Zhu Du}
	received the M.S. degree from the Chinese Academy of Sciences, Beijing, China, in 1982, and the Ph.D. degree from the University of California at Santa Barbara, Santa Barbara, CA, USA, in 1985, under the supervision of Prof. R. V. Book. Before settling at The University of Texas at Dallas, Richardson, TX, USA, he was a Professor with the Department of Computer Science and Engineering, University of Minnesota, Minneapolis, MN, USA. He was with the Mathematical Sciences Research Institute, Berkeley, CA, USA, for one year, with the Department of Mathematics, Massachusetts Institute of Technology, Cambridge, MA, USA, for one year, and with the Department of Computer Science, Princeton University, Princeton, NJ, USA, for one and a half years. Dr. Du is the Editor-in-Chief of the Journal of Combinatorial Optimization and is also on the editorial boards for several other journals.
\end{IEEEbiography}
\end{document}